\documentclass[journal]{IEEEtran}
 \usepackage{amsmath,amssymb}
 \usepackage{subfigure}
 \usepackage{graphicx,graphics,color,psfrag}
 \usepackage{cite,balance}
 \usepackage{caption}   
 \allowdisplaybreaks
 \usepackage{algorithm}
  \usepackage{accents}
   \usepackage{amsthm}
 \usepackage{bm} 
 \usepackage{algorithmic}
 \usepackage[english]{babel}
 \usepackage{multirow}
 \usepackage{enumerate}
 \usepackage{cases}
 \usepackage{stfloats}
 \usepackage{dsfont}
 \usepackage{color,soul}
 
\begin{document}
\newtheorem{theorem}{Theorem}
\newtheorem{acknowledgement}[theorem]{Acknowledgement}
\newtheorem{axiom}[theorem]{Axiom}
\newtheorem{case}[theorem]{Case}
\newtheorem{claim}[theorem]{Claim}
\newtheorem{conclusion}[theorem]{Conclusion}
\newtheorem{condition}[theorem]{Condition}
\newtheorem{conjecture}[theorem]{Conjecture}
\newtheorem{criterion}[theorem]{Criterion}
\newtheorem{definition}{Definition}
\newtheorem{exercise}[theorem]{Exercise}
\newtheorem{lemma}{Lemma}
\newtheorem{corollary}{Corollary}
\newtheorem{notation}[theorem]{Notation}
\newtheorem{problem}[theorem]{Problem}
\newtheorem{proposition}{Proposition}
\newtheorem{solution}[theorem]{Solution}
\newtheorem{summary}[theorem]{Summary}
\newtheorem{assumption}{Assumption}
\newtheorem{example}{\bf Example}
\newtheorem{remark}{\bf Remark}

\def\qed{$\Box$}
\def\QED{\mbox{\phantom{m}}\nolinebreak\hfill$\,\Box$}
\def\proof{\noindent{\emph{Proof:} }}
\def\poof{\noindent{\emph{Sketch of Proof:} }}
\def
\endproof{\hspace*{\fill}~\qed
\par
\endtrivlist\unskip}
\def\endproof{\hspace*{\fill}~\qed\par\endtrivlist\vskip3pt}

\def\E{\mathsf{E}}
\def\eps{\varepsilon}
\def\phi{\varphi}
\def\Lsp{{\boldsymbol L}}
\def\Bsp{{\boldsymbol B}}
\def\lsp{{\boldsymbol\ell}}
\def\Ltsp{{\Lsp^2}}
\def\Lpsp{{\Lsp^p}}
\def\Linsp{{\Lsp^{\infty}}}
\def\LtR{{\Lsp^2(\Rst)}}
\def\ltZ{{\lsp^2(\Zst)}}
\def\ltsp{{\lsp^2}}
\def\ltZt{{\lsp^2(\Zst^{2})}}
\def\ninN{{n{\in}\Nst}}
\def\oh{{\frac{1}{2}}}
\def\grass{{\cal G}}
\def\ord{{\cal O}}
\def\dist{{d_G}}
\def\conj#1{{\overline#1}}
\def\ntoinf{{n \rightarrow \infty }}
\def\toinf{{\rightarrow \infty }}
\def\tozero{{\rightarrow 0 }}
\def\trace{{\operatorname{trace}}}
\def\ord{{\cal O}}
\def\UU{{\cal U}}
\def\rank{{\operatorname{rank}}}
\def\acos{{\operatorname{acos}}}

\def\SINR{\mathsf{SINR}}
\def\SNR{\mathsf{SNR}}
\def\SIR{\mathsf{SIR}}
\def\tSIR{\widetilde{\mathsf{SIR}}}
\def\Ei{\mathsf{Ei}}
\def\l{\left}
\def\r{\right}
\def\({\left(}
\def\){\right)}
\def\lb{\left\{}
\def\rb{\right\}}

\setcounter{page}{1}

\newcommand{\eref}[1]{(\ref{#1})}
\newcommand{\fig}[1]{Fig.\ \ref{#1}}

\def\bydef{:=}
\def\ba{{\mathbf{a}}}
\def\bb{{\mathbf{b}}}
\def\bc{{\mathbf{c}}}
\def\bd{{\mathbf{d}}}
\def\bee{{\mathbf{e}}}
\def\bff{{\mathbf{f}}}
\def\bg{{\mathbf{g}}}
\def\bh{{\mathbf{h}}}
\def\bi{{\mathbf{i}}}
\def\bj{{\mathbf{j}}}
\def\bk{{\mathbf{k}}}
\def\bl{{\mathbf{l}}}
\def\bm{{\mathbf{m}}}
\def\bn{{\mathbf{n}}}
\def\bo{{\mathbf{o}}}
\def\bp{{\mathbf{p}}}
\def\bq{{\mathbf{q}}}
\def\br{{\mathbf{r}}}
\def\bs{{\mathbf{s}}}
\def\bt{{\mathbf{t}}}
\def\bu{{\mathbf{u}}}
\def\bv{{\mathbf{v}}}
\def\bw{{\mathbf{w}}}
\def\bx{{\mathbf{x}}}
\def\by{{\mathbf{y}}}
\def\bz{{\mathbf{z}}}
\def\b0{{\mathbf{0}}}

\def\bA{{\mathbf{A}}}
\def\bB{{\mathbf{B}}}
\def\bC{{\mathbf{C}}}
\def\bD{{\mathbf{D}}}
\def\bE{{\mathbf{E}}}
\def\bF{{\mathbf{F}}}
\def\bG{{\mathbf{G}}}
\def\bH{{\mathbf{H}}}
\def\bI{{\mathbf{I}}}
\def\bJ{{\mathbf{J}}}
\def\bK{{\mathbf{K}}}
\def\bL{{\mathbf{L}}}
\def\bM{{\mathbf{M}}}
\def\bN{{\mathbf{N}}}
\def\bO{{\mathbf{O}}}
\def\bP{{\mathbf{P}}}
\def\bQ{{\mathbf{Q}}}
\def\bR{{\mathbf{R}}}
\def\bS{{\mathbf{S}}}
\def\bT{{\mathbf{T}}}
\def\bU{{\mathbf{U}}}
\def\bV{{\mathbf{V}}}
\def\bW{{\mathbf{W}}}
\def\bX{{\mathbf{X}}}
\def\bY{{\mathbf{Y}}}
\def\bZ{{\mathbf{Z}}}

\def\mA{{\mathbb{A}}}
\def\mB{{\mathbb{B}}}
\def\mC{{\mathbb{C}}}
\def\mD{{\mathbb{D}}}
\def\mE{{\mathbb{E}}}
\def\mF{{\mathbb{F}}}
\def\mG{{\mathbb{G}}}
\def\mH{{\mathbb{H}}}
\def\mI{{\mathbb{I}}}
\def\mJ{{\mathbb{J}}}
\def\mK{{\mathbb{K}}}
\def\mL{{\mathbb{L}}}
\def\mM{{\mathbb{M}}}
\def\mN{{\mathbb{N}}}
\def\mO{{\mathbb{O}}}
\def\mP{{\mathbb{P}}}
\def\mQ{{\mathbb{Q}}}
\def\mR{{\mathbb{R}}}
\def\mS{{\mathbb{S}}}
\def\mT{{\mathbb{T}}}
\def\mU{{\mathbb{U}}}
\def\mV{{\mathbb{V}}}
\def\mW{{\mathbb{W}}}
\def\mX{{\mathbb{X}}}
\def\mY{{\mathbb{Y}}}
\def\mZ{{\mathbb{Z}}}

\def\cA{\mathcal{A}}
\def\cB{\mathcal{B}}
\def\cC{\mathcal{C}}
\def\cD{\mathcal{D}}
\def\cE{\mathcal{E}}
\def\cF{\mathcal{F}}
\def\cG{\mathcal{G}}
\def\cH{\mathcal{H}}
\def\cI{\mathcal{I}}
\def\cJ{\mathcal{J}}
\def\cK{\mathcal{K}}
\def\cL{\mathcal{L}}
\def\cM{\mathcal{M}}
\def\cN{\mathcal{N}}
\def\cO{\mathcal{O}}
\def\cP{\mathcal{P}}
\def\cQ{\mathcal{Q}}
\def\cR{\mathcal{R}}
\def\cS{\mathcal{S}}
\def\cT{\mathcal{T}}
\def\cU{\mathcal{U}}
\def\cV{\mathcal{V}}
\def\cW{\mathcal{W}}
\def\cX{\mathcal{X}}
\def\cY{\mathcal{Y}}
\def\cZ{\mathcal{Z}}
\def\cd{\mathcal{d}}
\def\Mt{M_{t}}
\def\Mr{M_{r}}
\def\O{\Omega_{M_{t}}}
\newcommand{\figref}[1]{{Fig.}~\ref{#1}}
\newcommand{\tabref}[1]{{Table}~\ref{#1}}

\newcommand{\var}{\mathsf{var}}
\newcommand{\fb}{\tx{fb}}
\newcommand{\nf}{\tx{nf}}
\newcommand{\BC}{\tx{(bc)}}
\newcommand{\MAC}{\tx{(mac)}}
\newcommand{\Pout}{p_{\mathsf{out}}}
\newcommand{\nnn}{\nn\\}
\newcommand{\FB}{\tx{FB}}
\newcommand{\TX}{\tx{TX}}
\newcommand{\RX}{\tx{RX}}
\renewcommand{\mod}{\tx{mod}}
\newcommand{\m}[1]{\mathbf{#1}}
\newcommand{\td}[1]{\tilde{#1}}
\newcommand{\sbf}[1]{\scriptsize{\textbf{#1}}}
\newcommand{\stxt}[1]{\scriptsize{\textrm{#1}}}
\newcommand{\suml}[2]{\sum\limits_{#1}^{#2}}
\newcommand{\sumlk}{\sum\limits_{k=0}^{K-1}}
\newcommand{\eqhsp}{\hspace{10 pt}}
\newcommand{\tx}[1]{\texttt{#1}}
\newcommand{\Hz}{\ \tx{Hz}}
\newcommand{\sinc}{\tx{sinc}}
\newcommand{\tr}{\mathrm{tr}}
\newcommand{\diag}{\mathrm{diag}}
\newcommand{\MAI}{\tx{MAI}}
\newcommand{\ISI}{\tx{ISI}}
\newcommand{\IBI}{\tx{IBI}}
\newcommand{\CN}{\tx{CN}}
\newcommand{\CP}{\tx{CP}}
\newcommand{\ZP}{\tx{ZP}}
\newcommand{\ZF}{\tx{ZF}}
\newcommand{\SP}{\tx{SP}}
\newcommand{\MMSE}{\tx{MMSE}}
\newcommand{\MINF}{\tx{MINF}}
\newcommand{\RC}{\tx{MP}}
\newcommand{\MBER}{\tx{MBER}}
\newcommand{\MSNR}{\tx{MSNR}}
\newcommand{\MCAP}{\tx{MCAP}}
\newcommand{\vol}{\tx{vol}}
\newcommand{\ah}{\hat{g}}
\newcommand{\tg}{\tilde{g}}
\newcommand{\teta}{\tilde{\eta}}
\newcommand{\heta}{\hat{\eta}}
\newcommand{\uh}{\m{\hat{s}}}
\newcommand{\eh}{\m{\hat{\eta}}}
\newcommand{\hv}{\m{h}}
\newcommand{\hh}{\m{\hat{h}}}
\newcommand{\Po}{P_{\mathrm{out}}}
\newcommand{\Poh}{\hat{P}_{\mathrm{out}}}
\newcommand{\Ph}{\hat{\gamma}}
\newcommand{\mat}[1]{\begin{matrix}#1\end{matrix}}
\newcommand{\ud}{^{\dagger}}
\newcommand{\C}{\mathcal{C}}
\newcommand{\nn}{\nonumber}
\newcommand{\nInf}{U\rightarrow \infty}



\setlength{\belowcaptionskip}{-10pt}

\title{\huge Energy Efficient Mobile Cloud Computing \\ Powered by Wireless Energy Transfer} 
\author{Changsheng You, Kaibin Huang and Hyukjin Chae     \thanks{\noindent C. You and K. Huang are  with the Dept. of EEE at The  University of  Hong Kong, Hong Kong (Email: csyou@eee.hku.hk, huangkb@eee.hku.hk). H. Chae is with LG Electronics,  S. Korea (Email: hyukjin.chae@lge.com). Updated on \today.}
}

\maketitle
\begin{abstract} Achieving long battery lives or even self sustainability has been a long standing challenge for designing mobile devices. This paper presents a novel solution that seamlessly integrates two technologies, \emph{mobile cloud computing} and \emph{microwave power transfer} (MPT), to enable computation in passive low-complexity devices such as sensors and wearable computing devices. Specifically, considering a single-user system, a base station (BS) either transfers power to or offloads computation from a mobile  to the cloud; the mobile uses harvested energy to compute given data either locally or by offloading.  A framework for energy efficient computing  is proposed that comprises a set of policies   for controlling  CPU cycles for the mode of local computing, time division between MPT and offloading for the other mode of offloading, and mode selection. Given the CPU-cycle statistics information and channel state information (CSI), the policies aim at maximizing  the probability of successfully computing given  data, called \emph{computing probability}, under the energy harvesting and deadline constraints. The policy optimization is translated into the equivalent problems of minimizing the mobile energy consumption for local computing and maximizing the mobile energy savings for offloading which are solved using  convex optimization theory. The structures of the resultant policies are characterized in closed form. Furthermore, given \emph{non-causal} CSI, the said analytical framework is further developed to support computation load allocation over multiple channel realizations, which further increases the computing probability. Last, simulation demonstrates the feasibility of wirelessly powered mobile cloud computing and the gain of its optimal control. 

\end{abstract}
\begin{IEEEkeywords}
Wireless power transfer, energy harvesting communications, mobile cloud computing, energy efficient computing.
\end{IEEEkeywords}

\vspace{+5pt}
\section{Introduction}

The explosive growth of Internet of Things (IoT) and mobile communication is  leading to  the deployment of tens of billions of {\color{black}{cloud-based}} mobile sensors and wearable computing devices in near future \cite{Melanie:Sensor:2012}. Prolonging their  battery lives and enhancing their computing capabilities are two key design challenges. They can be tackled by several  promising technologies: 1) \emph{microwave power transfer}  (MPT)  for powering  the mobiles  using microwaves\cite{Brown:RadioWPTHistory:1984}, 2) \emph{mobile computation offloading} (MCO) for offloading  computation-intensive tasks from the mobiles  to the cloud  \cite{kumar2013survey}, and 3) energy efficient \emph{local computing} using the mobile CPU. These  technologies are seamlessly integrated in the current work to develop a novel design framework for realizing wirelessly powered mobile cloud computing under the criterion of  maximizing the probability of successfully computing given data, called \emph{computing probability}.   {\color{black}{The framework is feasible since MPT has been proven in various experiments for powering small devices such as sensors \cite{Le:FarEgyHar:2008} 
or even small-scale airplanes and helicopters \cite{Shinohara:2014:WPT}. Furthermore, sensors and wearable computing devices targeted in the framework are expected to be connected by the cloud-based IoT in the future \cite{Melanie:Sensor:2012}, providing a suitable platform for realizing MCO.}} 
\vspace{+5pt}
\subsection{Prior Work}
MCO has been an active research area in computer science \cite{kumar2013survey} where research has focused on designing mobile-cloud systems and software architectures \cite{cuervo:MAUI:2010, Kosta:Thinkair:2012},
 virtual machine migration design in the cloud \cite{Gkatzikis:Migrate:2013} and code partitioning  techniques in the mobiles \cite{cuervo:MAUI:2010} 
 for reducing the energy consumption and improving the computing performance of mobiles. 
Nevertheless, implementation of MCO requires  data transmission and message passing  over wireless channels, incurring transmission power consumption \cite{kumar:CanOffloadSave2010}. The existence of such a tradeoff  has motivated cross-disciplinary research on jointly designing MCO and adaptive transmission algorithms to maximize the  mobile energy savings \cite{huang:DynamicOffload:2012, Barbarossa:MobileCloud:2014, ZhangY:OfflConnec:2015}.  A stochastic control algorithm was proposed in \cite{huang:DynamicOffload:2012} for adapting the offloaded components  of  an application to a time-varying wireless channel.  Furthermore, multiuser computation offloading  in a multi-cell system was explored in \cite{Barbarossa:MobileCloud:2014}, where the radio and computational resources were jointly allocated  for maximizing the energy savings under the latency constraints. {\color{black}{In \cite{ZhangY:OfflConnec:2015}, the threshold-based offloading policy was derived for the system with intermittent connectivity between the mobile and cloud.}}

Energy-efficient mobile (local) computing is also an active field where rich and diversified techniques have been designed for reducing the mobile energy consumption \cite{Yao:Scheduling:1995,Benini1:DPM:1999,Pillai:DVS:2001,lorch:DynamicVoltage:2001, yuan:RealCpuScheduling:2003, zhang:MobileMmodel:2013}. The scheduling of multiple computing tasks was optimized in \cite{Yao:Scheduling:1995} to  increase the energy savings. For the same objective, dynamic power management  was proposed in \cite{Benini1:DPM:1999} where components of computing tasks were dynamically reconfigured and selectively turned off. Another energy-efficient approach is to control the CPU-cycle frequencies under the deadline constraint by exploiting the fact that lowering the frequencies  reduces the energy consumption \cite{Pillai:DVS:2001,lorch:DynamicVoltage:2001, yuan:RealCpuScheduling:2003}.  Recently, the CPU-cycle frequencies are jointly controlled with MCO given a stochastic wireless channel in \cite{zhang:MobileMmodel:2013}. The framework is further developed in the current work to include the new feature of MPT.  This introduces several new design challenges. Among others, the algorithmic design of  local computing and offloading becomes more complex  under the energy harvesting constraint due to MPT, which  prevents energy consumption from exceeding the amount of harvested energy at every time instant  \cite{OzelUlukus:TransEnergyHarvestFading:OptimalPolicies:2011}.  Another challenge is that MPT and offloading time shares the mobile antenna and the time division has to be optimized. 

The MPT technology has  been developed for point-to-point high power transmission in the past decades  \cite{Brown:RadioWPTHistory:1984}. Recently, the technology is being further developed to power wireless communications. This has resulted in the emergence of an active field called \emph{simultaneous wireless information  and power  transfer} (SWIPT). Various SWIPT techniques have been developed by  integrating MPT with communication techniques including  MIMO transmission \cite{Zhang:MIMOBCWirelessInfoPowerTransfer},   OFDMA \cite{NgLo:MultiuserOFDMSInfoPowerTransfer}, two-way communication \cite{popovski:InteractiveTransferInformation:2013}  and relaying \cite{nasir:RelayHarveing:2013}. 
Furthermore,   existing wireless networks such as cognitive radio \cite{LeeZhangHuang:OppEnergyHarvestingCognitiveRadio:2013} and cellular networks \cite{HuangLauArXiv:EnablingMPTinCellularNetworks:2013} have been redesigned to feature  MPT.   Recent advancements on SWIPT are surveyed in \cite{bi:PowerCommunication:2014} and \cite{Huang:CuttingLastWiress:2014}. Most prior work on SWIPT aims at optimizing communication techniques to maximize the MPT efficiency and system throughput. In contrast, the current work focuses on optimizing the local computing and offloading under a different design criterion of maximum computing probability. 
\vspace{+10pt}
\subsection{Contributions and Organization}
Consider a single-user system comprising  one multi-antenna \emph{base station} (BS) using transmit/receive beamforming for transferring  power  to a single-antenna mobile or  relaying offloaded data from the mobile to the cloud. To compute a fixed amount of data, the mobile operates in one of the two available modes:  local computing and offloading. In the mode of local computing, MPT occurs simultaneously as computing  based on the controllable CPU-cycle frequencies. Nevertheless, in the mode of offloading, the given computation duration is adaptively partitioned for separate MPT and offloading since they  share the mobile antenna. Assume that the mobile has the knowledge of statistics information of CPU cycles and channel state information (CSI).  The individual modes as well as mode selection are optimized for maximizing the computing probability under the energy harvesting and deadline constraints.  For tractability, the metric is transformed into equivalent ones, namely \emph{average mobile energy consumption} and \emph{mobile energy savings}, for the modes of local computing and offloading, respectively. {\color{black}{Compared with \cite{zhang:MobileMmodel:2013}}, the current work integrates MPT with the mobile cloud computing, which introduces new theoretical challenges. In particular, the  energy harvesting constraint arising from MPT makes the optimization problem for local computing non-convex. To tackle the challenge, the convex relaxation technique is applied without compromising the optimality of the solution. It is shown in the sequel that the local computing policy of \cite{zhang:MobileMmodel:2013} is a special case of the current work where the transferred power is sufficiently high.  Furthermore, the case of dynamic channel for mobile cloud computing is explored. Approximation methods are used for deriving the simple and close-to-optimal policies.}

 The contributions of the current work are summarized as follows. 

\begin{itemize}
\item \emph{Optimal local computing}: First, consider a static channel.  For the mode of local computing,  CPU-cycle frequencies are optimized for minimizing the average mobile energy consumption under the energy harvesting and deadline constraints. The corresponding  optimization problem is non-convex but solved by the convex relaxation that is proved to maintain the  optimality. The resultant policy is shown to depend on two derived thresholds on the transferred power. The first determines the feasibility of successful computing.  Given feasible computing,  the optimal frequency of the $n$-th CPU cycle is shown to be \emph{proportional} to $p_k^{-\frac{1}{3}}$ with $p_k$ being the probability of its occurrence  and independent of the transferred power if it is above the second threshold; otherwise, the frequency is \emph{proportional} to $(p_k + \lambda)^{-\frac{1}{3}}$ with $\lambda$ being a constant determined by the transferred power. 

\item \emph{Optimal computation  offloading}: Next, for the mode of  offloading, the partition of the computation duration for  separate MPT and offloading is optimized to maximize the energy savings. As a result, a threshold is derived on the product of the BS transmission power and squared channel power gain, above which offloading is feasible. Given feasibility, the optimal offloading duration is shown to be \emph{proportional} to the input data size and \emph{inversely proportional} to the channel bandwidth.   In other words, small data size and large bandwidth reduces time allocated for offloading and increases time for MPT and vice versa. 

\item \emph{Mobile mode selection}:  The above results are combined to select the mobile mode for maximizing the computing probability. Given feasible computing in both modes,  the one yielding the larger energy savings is preferred and the selection criterion is derived in terms of thresholds on the BS transmission power as well as the deadline for computing.

\item \emph{Optimal data   allocation for a dynamic channel}:   Last, the above results are extended to the case of a dynamic channel, modeled as independent and identically distributed (i.i.d.) block fading,  and  \emph{non-causal} CSI at the mobile (acquired from e.g., channel prediction). The problem of optimizing  an individual  mobile mode (local computing or offloading) is formulated based on the \emph{master-and-slave model} using the same metric as the fixed-channel counterpart. The master problem concerns the  optimal data allocation  for computing in a fixed number of fading blocks. The slave problem targets the mode optimization in a single fading block for allocated data and is similar to the fixed-channel counterpart. By approximating the master problems, sub-optimal  policies are designed for data allocation and shown by simulation to be close-to-optimal. The results can be straightforwardly combined to enable the mode selection. 

\end{itemize}

The remainder of this paper is organized as follows. The system model is introduced in Section II. Section III presents the optimal policies  for mobile mode optimization and selection for the case of static channel. The results are extended in 
Section IV to the case of dynamic channel. Simulation results are presented in  Section V, followed by the conclusion in
 Section VI.

\section{System Model}
\begin{figure*}[t]
\begin{center}
\subfigure[Wirelessly powered mobile cloud computing system]{\includegraphics[width=12cm]{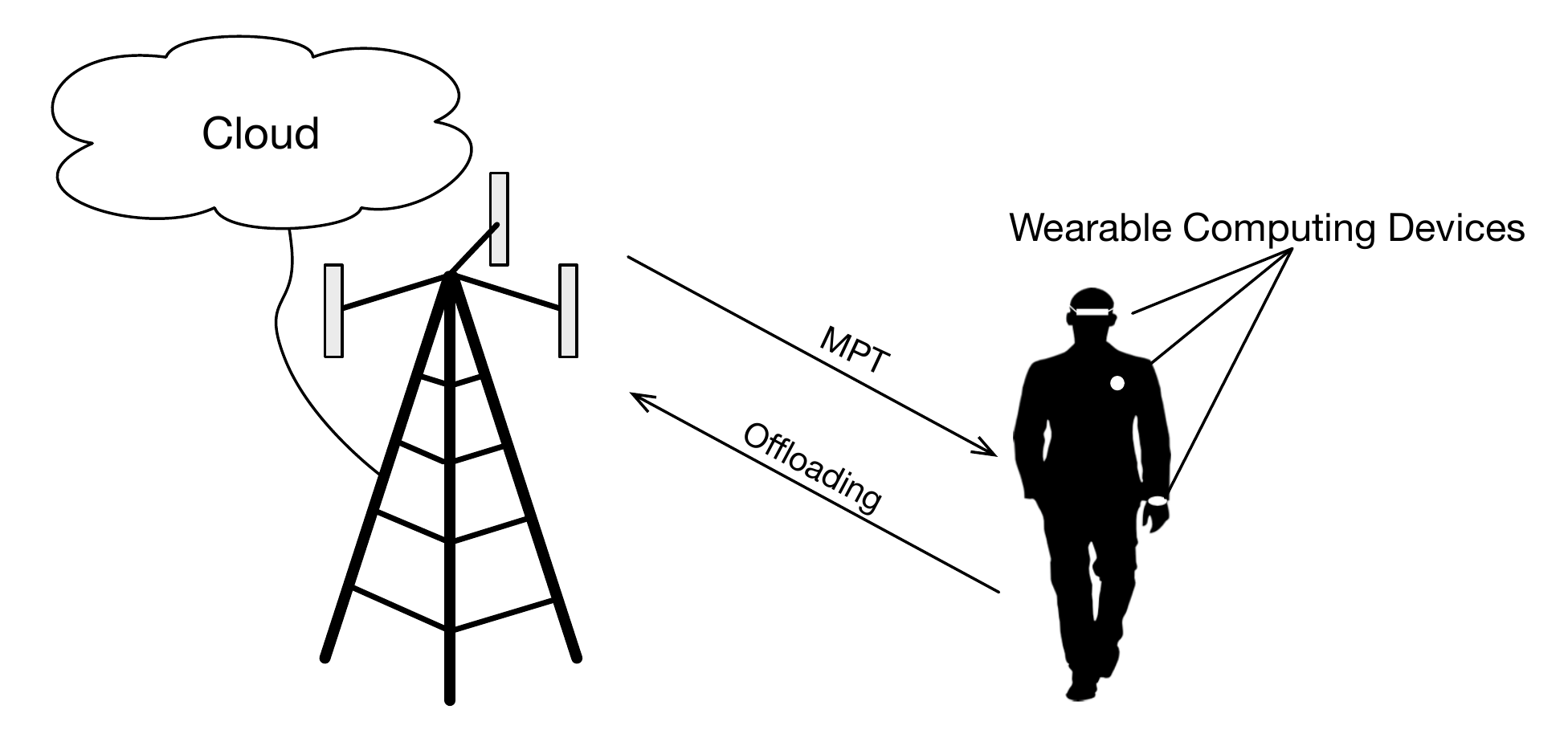}}\\
\subfigure[Mobile operation modes]{\includegraphics[width=11cm]{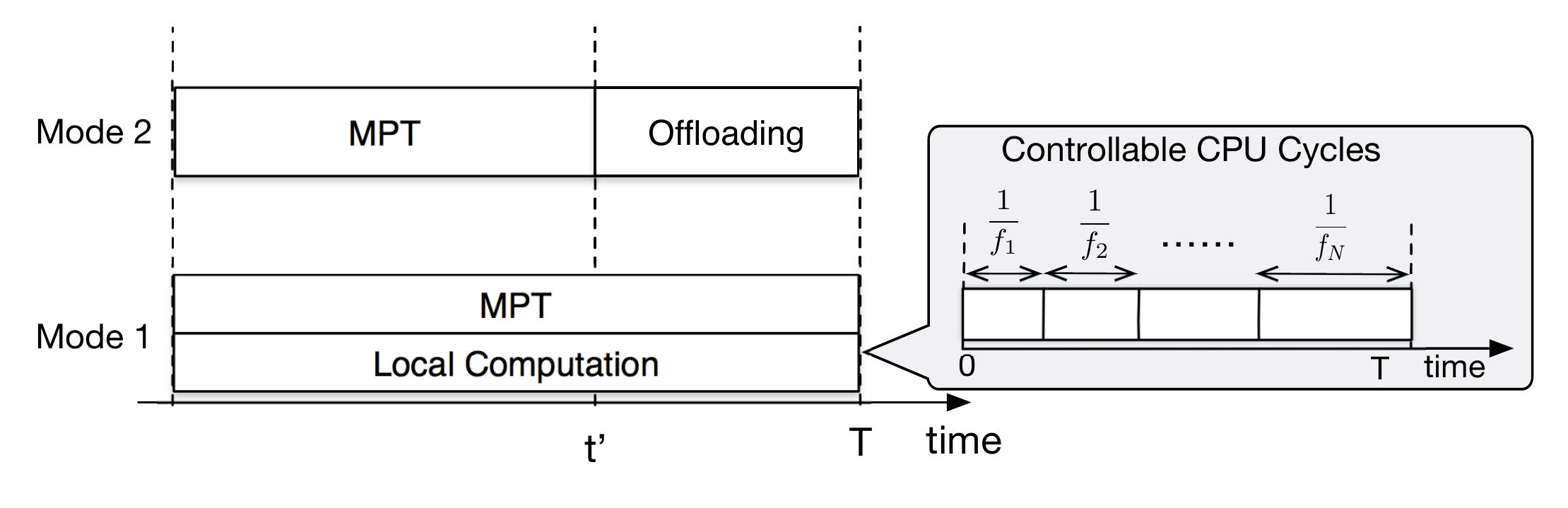}}\\
\caption{(a) Wirelessly powered mobile cloud computing system and (b) the mobile operation modes.}
\label{system_model}
\end{center}
\end{figure*}

{\color{black}{Consider a wirelessly powered mobile cloud computing  system model shown in Fig.~\ref{system_model}(a) comprising one  single-antenna mobile, e.g., mobile sensor and wearable computing device with milliwatt power consumption, and one multi-antenna BS that is a node of a cloud.}} The BS either transfers power wirelessly to  or offloads computation from the mobile.
{\color{black}{Consider the computation of a single task that cannot be split for partial local computing and partial offloading. In other words, the mobile operation mode is either local computing or offloading but not both (see Fig.~\ref{system_model}(b)).}} The local computing and MPT can be performed simultaneously while the offloading and MPT are non-overlapping in time assuming half-duplex transmission. {\color{black}{The current work can be generalized to the case of computing a multi-task program which allows program partition and thus the simultaneous operation of these two modes.}} The multi-antenna BS uses beamforming to transfer power or receive the signal. {\color{black}{Assuming channel reciprocity \cite{Zhang:MIMOBCWirelessInfoPowerTransfer, NgLo:MultiuserOFDMSInfoPowerTransfer}, the effective scalar  channel power gain is represented as $h$}}. Both cases of static and dynamic channel are considered. The dynamic channel is modeled as i.i.d. block fading where the channel power gain is fixed in each fading block and i.i.d. over different blocks. Let $P_b$ denote the BS transmission power. It is assumed that both  $P_b$ and $h$ are accurately estimated\footnote{\color{black}{Accounting for imperfect channel knowledge, robust optimization techniques can be applied to the current framework for obtaining control policies to guarantee a required computing probability.}} by the mobile and used for controlling the local computing and offloading.
 {\color{black}{During MPT, assuming the mobile has infinite battery storage\footnote{\color{black}{Considering finite battery  storage at mobiles will complicate the corresponding structure of the optimal solution. However, it is expected to have the threshold-based structure as we derive for the case of infinite storage.}}, since the energy harvested from the noise is negligible, the energy harvested by the mobile in one time unit is $\upsilon P_b h$ where the constant $0 < \upsilon \leq 1 $ represents the energy conversion efficiency. }}

\subsection{Local Computing Model} \label{Section:LC:Model}
Adopting  the model in   \cite{yuan:RealCpuScheduling:2003,  zhang:MobileMmodel:2013}, the required number of CPU cycles for computing given data is randomly generated as follows. Let $L$ denote the number of bits in input data  for computation and $T$ represent its deadline.  Define the \emph{CPU cycle information} (CCI) as the number of CPU cycles required for computing $1$-bit of input data. Then the CCI  can be modeled as a random variable denoted as $X$ and $LX$ gives the number of CPU cycles for computing $L$-bit input data \cite{lorch:DynamicVoltage:2001}.   Define $N_0$ as a positive integer such that $\Pr(X > N_0) \leq \epsilon$ where $\epsilon$ is a small real number: $0 < \epsilon \ll 1$. It follows that $\Pr(LX > N)\leq \epsilon$ where $N = L N_0$. Then given $L$-bit input data, $N$ upper bounds the number of CPU cycles almost surely. {\color{black}{Define $p_k = \Pr(LX \geq k)$ with $k = 1, 2, \cdots, N$ such that $p_k$ is the probability that the data has not been processed completely after $k$ CPU cycles.}} In other words, $p_k$ is the probability that the $k$-th cycle is executed. {\color{black}{Note that the sequence $p_1, p_2, \cdots, p_N$ is monotone decreasing, regardless of any probability distribution of the random variable $X$}} . 

 \begin{assumption}\label{AS:CCI}\emph{The mobile only has the knowledge of the CCI distribution.}\footnote{\color{black}{The results on CPU-cycle optimization in the current work  can be straightforwardly extended to the case where the mobile has the prior knowledge of CCI instead of CCI distribution, resulting in a simpler policy structure.}} 
\end{assumption}

The energy consumption of local computing is modeled as follows.  In practice,  mobile energy consumption is contributed by computation, transmission and fixed circuit power. The circuit power is omitted for simplicity but it can be accounted for by adding a constant in the problem formulation that, however, does not affect the solution method and key results. Let $E(f)$ denote the amount of energy for computation in a single CPU cycle with the frequency $f$ (or equivalently having a CPU-cycle duration of $1/f$). Following the model in \cite{Robert:CMOS:1992, burd:CpuEnergy:1996}, under the assumption of a low CPU voltage,  $E(f)=\gamma f^2$ where $\gamma$ is a constant determined by  the  switched capacitance. {\color{black}{However, the analysis of current work also applies to a  general polynomial function $E(f)$ that is monotone increasing and convex.}} Let $f_1, f_2, \cdots, f_N$ denote the CPU-cycle frequencies for CPU cycles $1, 2, \cdots, N$, 
respectively. 

Last, for tractability, several assumptions are made for the case of dynamic channel. 
\begin{assumption}\label{AS:CSI}\emph{For the case of dynamic channel, the mobile has \emph{non-causal} CSI.}\footnote{\color{black}{Considering causal CSI, if the channel gains of adjacent fading blocks are temporally correlated that can be modeled as the channel state Markov chain, dynamic programming  can be used for deriving the optimal solution, which, however, is intractable to analyze and thus  is not considered in the current work.}}

 \end{assumption}
 
Non-causal CSI in Assumption~\ref{AS:CSI} corresponds to the case where the mobile predicts the channel variation or has a pre-determined  trajectory.  The assumption allows off-line data allocation for computing in multiple fading blocks.

\begin{assumption}\label{As:Channel}\emph{For the case of  dynamic channel, the number of channel realizations in the computation duration $[0, T]$ is much smaller than the typical number of CPU cycles required for mobile local computing in the same duration.}
\end{assumption}

Assumption~\ref{As:Channel} states different time scales for the channel variation and CPU cycles, allowing energy-efficient control policies to be designed using a decomposition approach in the sequel.

\begin{assumption}\label{As:Data}\emph{For  the case of dynamic channel, the input data can be divided continuously into parts that can be computed separately.}
\end{assumption}

 {\color{black}{The input data of several applications can be divided, such as Gzip compression and feature extraction. Assumption~\ref{As:Data} made for tractability simplifies the data allocation  for computing in separate fading blocks to reduce energy consumption. In practice, the optimal solution in the current paper can be discretized by rounding. }}
\subsection{Computation Offloading Model } 
Consider computation offloading where the mobile transmits the data to  the BS for computing in the cloud and receives the result via the BS. Given the mobile transmission power $P_t$, the uplink channel capacity (in bit/s), denoted as $C$,  is given by:  
\begin{equation}
 C=B\log\left(1+ \frac{P_t h}{\sigma^2}\right)  \label{rate} \nn
\end{equation} 
where $B$ is the channel  bandwidth and  $\sigma^2$ is the variance of complex white Gaussian channel  noise. Moreover, it is assumed that the time for computing in the cloud and transmitting the computation result from the BS to the mobile is negligible since the cloud has practically infinite computational resources and the BS can afford high transmission power to reduce the downlink transmission delay. Last, the computation result is assumed of a small size such that demodulating the  result data at the mobile results in negligible energy consumption compared with that for local computing 
or offloading.

\subsection{Performance Metrics} \label{Section:Metric}
Given the CCI $X$, let $E_{\textrm{MPT}}(t)$  represent  the amount of cumulative energy harvested by the mobile over duration $[0, t]$ and   $E_{\textrm{mob}}(t, X, \{f_k\}_{k=1}^X)$ as the amount of  cumulative energy consumed by the mobile  for local computing or offloading. {\color{black}{Then the performance metric, computing probability, which is the probability of successfully computing given data, is denoted as $P_c$}}, defined as 
\begin{equation}\label{Eq:Pc:Def:LC}
P_{\textrm{c}}\!=\!\mathbb{E}\l [I\l(E_{\textrm{MPT}}(t) \!\ge\! E_{\textrm{mob}}(t, \!X, \{f_k\}_{k=1}^X), \ \forall \ t \!\in\! [0, T] \r)\r ]
\end{equation}
where the indicator function $\emph{I}(\mathcal{E})$ gives $1$ when the event $\mathcal{E}$ occurs and $0$ otherwise. 
Note that the argument of the indicator function in \eqref{Eq:Pc:Def:LC} specifies  the energy harvesting constraint.  It can be observed from \eqref{Eq:Pc:Def:LC} that maximizing the computing probability  by optimizing the CPU-cycle frequencies is equivalent to minimizing the mobile energy consumption $E_{\textrm{mob}}(t, X, \{f_k\}_{k=1}^X)$ for $t\in [0, T]$. Nevertheless, this is infeasible since $X$ is unknown to the mobile based on Assumption~\ref{AS:CCI}. {\color{black}{To overcome this difficulty, the current work instead focuses on  minimizing the \emph{expected} energy consumption in the duration $[0, T]$, namely $\mathbb{E}\left[E_{\textrm{mob}}(T, X, \{f_k\}_{k=1}^X)\r]$, but still under the same constraints, which is given by
\begin{equation}\label{Eq:local:objective:ave_egy}
\mathbb{E}\left[E_{\textrm{mob}}(T, X, \{f_k\}_{k=1}^X)\r]=\sum_{k=1}^{N} \gamma p_k f_k^2+\sum_{k=N+1}^{LX} \gamma p_k f_k^2.
\end{equation}
Furthermore, for a small $\epsilon$ that is close to zero (see Section~\ref{Section:LC:Model}), the second component of \eqref{Eq:local:objective:ave_egy} is negligible. Then the following approximation is used for tractability:
\begin{equation}\label{Eq:local:objective}
\mathbb{E}\left[E_{\textrm{mob}}\l(T, X, \{f_k\}_{k=1}^X\r)\r] \approx  \sum_{k=1}^{N} \gamma p_k f_k^2.
\end{equation}
}}

Next, consider offloading. By slight abuse of notation, let $E_{\textrm{MPT}}(t')$  represent  the amount of total energy harvested by the mobile over duration $[0, t']$ and   $E_{\textrm{mob}}(t')$ as the amount of  energy consumed by the mobile  for offloading in the duration $[t', T]$. Then the computing probability for the current case can be  defined as 
\begin{equation}\label{Eq:Pc:Def:OL}
P_{\textrm{c}}=\mathbb{E}\l [I\l(E_{\textrm{MPT}}(t') \ge E_{\textrm{mob}}(t'), \ \forall \ t'\in [0, T] \r)\r ].
\end{equation}
As observed from \eqref{Eq:Pc:Def:OL},  maximizing the computing probability is equivalent to maximizing the energy savings $\l[E_{\textrm{MPT}}(t')  - E_{\textrm{mob}}(t')\r]$.

 \section{Energy Efficient Mobile Cloud Computing  with a Static Channel}
In this section, given a static channel, the CPU-cycle frequencies and MPT-and-offloading time division are optimized for local computing and offloading, respectively. Then the results are combined for optimizing the mobile  mode selection.

 \subsection{Energy Efficient Local Computing with a Static Channel} 

 \subsubsection{Problem Formulation}  Based on the discussion in Section~\ref{Section:Metric}, the problem of optimizing CPU-cycle frequencies aims at minimizing the average mobile energy consumption in \eqref{Eq:local:objective} under two constraints. The first is the deadline constraint: $\sum_{k=1}^{N} \frac{1}{f_k} \le T$. The second is the energy harvesting constraint comprising of  $N$ sub-constraints given as 
 \begin{equation}
\sum_{k=1}^m \gamma f_k^2 \leq \upsilon P_b h\sum_{k=1}^m \frac{1}{f_k},\quad   m=1, 2, \cdots, N 
\end{equation}
where the left-hand side of the inequality is the total energy consumed  by the first $m$ CPU cycles and the right-hand side is the total energy harvested by the end of $m$-th cycle.
 It follows that the  optimization problem is formulated as: 
\begin{equation}
\begin{aligned}
 \min_ {\{f_k\} } \quad  & \sum_{k=1}^{N} \gamma p_k f_k^2 \\
\text{s.t.}\quad 
& \sum_{k=1}^m \gamma  f_k^2 \leq \upsilon P_b h\sum_{k=1}^m \frac{1}{f_k}, & m = 1, 2, \cdots, N,\\ 
& \sum_{k=1}^{N} \frac{1}{f_k}\le T,  \\
&  f_k >0,  & k = 1, 2, \cdots, N. 
\end{aligned}
\tag{$\textbf{P1}$} 
\end{equation}
{\color{black}{Note that $p_k$ is included in the objective function to formulate the average energy consumption, while the energy harvesting constraint is defined for all possible CPU-cycle realizations and thus is without $p_k$.}}
\subsubsection{Solution} It can be observed that the energy harvesting sub-constraints in Problem P1 are non-convex, resulting in a non-convex optimization problem. 
To address this issue, first, define a set of new variables $\{y_k\}$ as  $y_k= \frac{1}{f_k}$ for all $k$. Substituting them into Problem P1 and furthermore relaxing the equality constraint   $y_k f_k= 1$ to be $y_k f_k \ge 1$ gives 
\begin{equation}
\begin{aligned}
 \min_{\{f_k, y_k\} }  \quad  & \sum_{k=1}^{N} \gamma p_k f_k^2 \\
\text{s.t.}\quad 
& \sum_{k=1}^m \gamma  f_k^2 \leq \upsilon P_b h\sum_{k=1}^m y_k, & m = 1, 2, \cdots, N,\\ 
& \sum_{k=1}^{N} y_k\le T,  \\
&  f_k >0, \quad \frac{1}{f_k}-y_k \le 0,  & k = 1, 2, \cdots, N. 
\end{aligned}
\tag{$\textbf{P2}$} 
\end{equation}
Observed that Problem P2 is a convex optimization problem. Nevertheless, the  relaxation mentioned earlier has no effect on the solution optimality as shown in the following lemma. 
\begin{lemma}\label{Lem:Relax} \emph{The solution for  Problem P2 also solves P1.}
\end{lemma}
\begin{proof}
See Appendix~\ref{App:Relax}.
\end{proof}

  Lemma~\ref{Lem:Relax} is essential for solving the non-convex Problem P1 by equating it with  the convex Problem P2 that yields the same solution but  admits powerful convex optimization algorithms. 

To characterize the structures of the optimal CPU-cycle frequencies, several useful properties of the solution for Problem P2 are given as follows. 

\begin{lemma}\label{Lem:P2}\emph{The solution for Problem P2, denoted as $\{y^*_k, f_k^*\}$,  satisfies the following. 
\begin{enumerate}
\item The deadline constraint is active:\\ $\sum_{k=1}^N y_k^* = \sum_{k=1}^N \frac{1}{f^*_k} =T$.
\item The optimal CPU-cycle frequencies can be written as 
 \begin{equation}
f_k^*=\left[\frac{\mu-\upsilon P_b h \left(\sum_{m=k}^N \lambda_m\right)}{2 \gamma \left(p_k+\sum_{m=k}^N \lambda_m\right)}\right] ^{\frac{1}{3}}, 
 \quad ~~\forall k \label{Eq:Opti:f} 
\end{equation}
where the nonnegative variables $\mu$, $\{\lambda_m\}$ are the Lagrange multipliers associated with the deadline and energy harvesting constraints, respectively.
\item Furthermore, $f_1^* < f_2^*\cdots <  f_N^*$.
\end{enumerate}
}
\end{lemma}

\begin{proof}
See Appendix~\ref{App:P2}.
\end{proof}
Another important property is stated in the following lemma. 
\begin{lemma}\label{Lem:RelaxMore}\emph{The solution for Problem P2 can also be derived by solving the following Problem P3 that results from removing the first $(N-1)$ energy harvesting sub-constraints of P2: 
\begin{equation}
\begin{aligned}
 \min_{\{f_k, y_k\} }  \quad  & \sum_{k=1}^{N} \gamma p_k f_k^2 \\
\text{s.t.}\quad 
& \sum_{k=1}^N \gamma  f_k^2 \leq \upsilon P_b h\sum_{k=1}^N y_k,\\ 
& \sum_{k=1}^{N} y_k\le T,  \\
&  f_k >0, \quad \frac{1}{f_k}-y_k \le 0, & k = 1, 2, \cdots, N.
\end{aligned}
\tag{$\textbf{P3}$} 
\end{equation}
}\end{lemma}
\begin{proof}
See Appendix~\ref{App:RelaxMore}.
\end{proof}

{\color{black}{Lemma~\ref{Lem:RelaxMore} means the Lagrange multipliers associated with the first $(N-1)$ energy harvesting sub-constraints are equal to zero, i.e., $\lambda_k=0$ for $k=1, 2, \cdots N-1$. }} Combining Property 2) in Lemma~\ref{Lem:P2} and Lemma~\ref{Lem:RelaxMore} simplifies the expression for the optimal CPU-cycle frequencies as 
\begin{equation}
f_k^*=\left[\frac{\mu- \upsilon P_b h \lambda}{2 \gamma (p_k+ \lambda)}\right] ^{\frac{1}{3}},  \quad  \forall k  \label{reduced optimal f} 
\end{equation}
where $\lambda_N$ is re-denoted as $\lambda$ for ease of notation. To obtain the closed-form expressions for $\{f_k^*\}$, define two positive constants $a$ and $a'$ as 
\begin{equation}\label{Eq:a and a'}
a = \frac{\gamma N^3}{\upsilon T^3} \quad \text{and}\quad a' = \frac{\gamma}{\upsilon  T^3}\left( \sum_{k=1}^{N}  p_k^{\frac{1}{3}}\right)^2 \left(\sum_{k=1}^{N}p_k^{-\frac{2}{3}}\right).
\end{equation} Then the main result of this subsection is stated as follows.

\begin{theorem}\label{Theo:LocalComp}\emph{(Optimal CPU-cycle frequencies for local computing)}. \emph{The optimal CPU-cycle frequencies $\{f^*_1, f^*_2, \cdots, f^*_N\}$ that solve the optimization problem P3 satisfy the following. 
\begin{enumerate}
\item If $ P_b h < a$, $\{f^*_1, f^*_2, \cdots, f^*_N\}$ is an  empty set since Problem P3 is infeasible. 
\item If $a \leq P_b h < a'$, 
\begin{equation}\label{Eq:CPU:lambda:nonzero}
f_k^*=\l[\frac{1}{T}\sum_{m=1}^{N} (p_m+\lambda)^{\frac{1}{3}} \r] (p_k + \lambda)^{-\frac{1}{3}}, \qquad \forall k 
\end{equation}
{\color{black}{where the positive constant $\lambda$ is the Lagrange multiplier with respect to (w.r.t) the energy harvesting constraint of P3}} and it satisfies
\begin{equation}
\left[ \sum_{k=1}^N  (p_k+\lambda) ^{\frac{1}{3}}\right]^2 \left[ \sum_{k=1}^N({p_k+\lambda})^{-\frac{2}{3}}\right]=\frac{\upsilon P_b h T^3}{\gamma}. \label{Eq:lambda:ph}
\end{equation}
\item If $P_b h \geq a'$, $\{f_k^*\}$ are independent of $P_b h$,
\begin{equation}\label{Eq:CPU:lambda:zero}
f_k^* =\l(\frac{1}{T}\sum_{m=1}^{N} p_m^{\frac{1}{3}} \r) p_k^{-\frac{1}{3}}, \qquad \forall k. 
\end{equation}
\end{enumerate}}
\end{theorem}

\begin{proof}
See Appendix~\ref{App:LocalComp}.
\end{proof}
The first case ($P_b h< a$) corresponds to the scenario where the transferred power $P_b h$ is so low that it is impossible to complete the local computing within the deadline $T$. In the second case ($a \leq P_b h < a'$), the transferred power is not large but sufficient  for meeting the deadline and consequently the optimal CPU-cycle frequencies are functions of $P_b h$. The transferred power is large in the last case ($P_b h \geq a'$) where increasing $P_b$ no longer has any effect on  the optimal CPU-cycle frequencies and only increases the amount of energy savings. In this case, the optimal CPU-cycle frequencies depend only on the CCI distribution and the deadline.

\begin{remark}\label{Rem:Trd:Loc}\emph{As observed from the definition of $a$ in \eqref{Eq:a and a'} and Theorem~\ref{Theo:LocalComp}, there exists a tradeoff between the deadline $T$ and BS transmission power $P_b$.  Specifically, meeting a stricter  deadline requires larger $P_b$ and vice versa. {\color{black}{In other words, the BS can control $P_b$ to increase the computing probability. However, jointly designing the control policies at the BS and mobile is challenging and $P_b$ is  assumed fixed for simplicity.}} }
 \end{remark}
 
 \begin{remark}[BS transmission power and the parameter $\lambda$]\emph{ Recall that $\lambda$ is the Lagrange multiplier for solving Problem P3. The optimal CPU-cycle frequencies in Cases 2) and 3) of Theorem~\ref{Theo:LocalComp} correspond to $\lambda > 0$ and $\lambda = 0$, respectively. The case of $\lambda \!=\! 0$ (or equivalently $P_b h \!\ge\! a'$) implies that the energy harvesting constraint in Problem P3 is inactive at the solution point, reducing the problem and solution to be identical to those in \cite{zhang:MobileMmodel:2013} considering local computing without MPT. 
}
 \end{remark}

\begin{corollary}[Minimum average energy consumption]\label{Cor:Egy:Local}\emph{It follows from Theorem~\ref{Theo:LocalComp} that the minimum average energy consumption for the local computing, denoted as $\bar{E}_{\textrm{loc}}^*$, is given as follows. 
\begin{enumerate}
\item If $a \leq P_b h < a'$, 
\begin{equation}
\bar{E}_{\textrm{loc}}^*\!=\!\frac{\gamma}{T^2} \left[ \sum_{k=1}^N  (p_k+\lambda) ^{\frac{1}{3}}\right]^2 \left[ \sum_{k=1}^N p_k ({p_k+\lambda})^{-\frac{2}{3}}\right] \label{Eq:Local:egy:lambda:nonzero}
\end{equation}
and $\bar{E}_{\textrm{loc}}^*$ is a monotone-decreasing function of $P_b h$ satisfying 
\begin{equation}
 \frac{\gamma}{T^2}\left( \sum_{k=1}^{N} p_k^\frac{1}{3}\right)^3   < \bar{E}_{\textrm{loc}}^*\leq \frac{\gamma {N}^2 }{T^2} \sum_{k=1}^{N} p_k.  \nn
\end{equation}
\item If $P_b h \geq a'$, 
\begin{equation}
\bar{E}_{\textrm{loc}}^* = \frac{\gamma}{T^2}\left( \sum_{k=1}^{N} p_k^\frac{1}{3}\right)^3 \label{Eq:Local:egy:lambda:zero}
\end{equation}
that is independent of $P_b h$. 
\end{enumerate}
Moreover,  the corresponding maximum average mobile energy savings, denoted as $\bar{S}_{\textrm{loc}}^*$,  is given as $\bar{S}_{\textrm{loc}}^* = \upsilon P_b h T - \bar{E}_{\textrm{loc}}^*$. 
}
\end{corollary}
\begin{proof}
See Appendix~\ref{App:Egy:Local}.
\end{proof}

\subsection{Energy Efficient Offloading with a Static Channel} 
This sub-section focuses on computation offloading. The time division between MPT and offloading is optimized for maximizing the mobile energy savings (see Section~\ref{Section:Metric}). 

\subsubsection{Problem Formulation} The objective function, namely the mobile energy savings, is obtained as follows.  As shown in Fig.~\ref{system_model}(b),  for current operation mode and $t' \in (0, T)$, the time interval $[0, T]$ is divided into two parts: $[0, t']$ and $(t', T]$, corresponding to MPT and offloading, respectively. Let the amount of energy harvested over the interval $[0, t']$ be defined as a function of $t'$: $E_{\textrm{MPT}}(t')=\upsilon P_b h t'$. Next, consider offloading in the interval $(t', T]$. Fixed-rate transmission over this interval is assumed since this is the most energy-efficient data transmission policy under a deadline constraint as proved in \cite{PrabBiyi:EenergyEfficientTXLazyScheduling:2001}. Then the energy consumption for offloading, denoted as $E_{\textrm{off}}(t')$, can be written as $E_{\textrm{off}}(t')=[2^\frac{L}{B(T - t')}-1]\frac{\sigma^2}{h} (T - t')$.  The energy savings is thus given as $ \l[E_{\textrm{MPT}}(t') - E_{\textrm{off}}(t')\r]$.  Varying $t'$ changes the harvested energy, offloading energy consumption as well as energy savings.  Specifically, as $t'$ increases,   $E_{\textrm{MPT}}(t')$ grows linearly with $t'$ but $E_{\textrm{off}}(t')$ monotonically increases as shown in\cite{PrabBiyi:EenergyEfficientTXLazyScheduling:2001}. Therefore, the energy savings may not be a monotone function and thus optimization is necessary. To simplify notation, define the offloading duration $t = T - t'$ and then $E_{\textrm{MPT}}(t')$ and $E_{\textrm{off}}(t')$ can be rewritten as  $E_{\textrm{MPT}}(t)=\upsilon P_b h (T - t)$ and $E_{\textrm{off}}(t)=(2^ \frac{L}{Bt}-1)\frac{\sigma^2}{h} t$. Substituting the expression of $E_{\textrm{MPT}}(t)$ and $E_{\textrm{off}}(t)$ rewrites the objective function as
{\color{black}{\begin{equation}
E_{\textrm{MPT}}(t)-E_{\textrm{off}}(t)\!=\!\upsilon P_b h T + \l(\frac{\sigma^2}{h}\!-\!\upsilon P_b h\r) t- \frac{\sigma^2}{h} t 2^{\frac{L}{Bt}}. \label{Eq:Offloading:Objective}
\end{equation}}}

Given this objective function, the problem for the current case is formulated as 
\begin{equation}
\begin{aligned}
 \max_t  \qquad  &E_{\textrm{MPT}}(t)-E_{\textrm{off}}(t)\\
\text{s.t.}\qquad 
 & 0< t < T,\\
 & E_{\textrm{MPT}}(t)-E_{\textrm{off}}(t) \ge 0.\nn
\end{aligned}
\tag{$\textbf{P4}$} 
\end{equation}

\subsubsection{Solution} Define \begin{equation}
\rho(h)  =\frac{\ln 2}{B\left[1+ W(\frac{\upsilon P_b h^2}{\sigma ^2 e} -\frac{1}{e})\right]}  \label{Eq:rho}
\end{equation}
where  $W(x)$ is the Lambert function defined as the solution for $W(x) e^{W(x)}=x$. Problem P4 is a convex problem as shown in the following lemma. 

\begin{lemma}[Convexity of P4] \label{Lem:Conv:Off}\emph{The objective function of Problem P4 is a concave function for $t\in (0, \infty)$ and maximized at $t = \rho(h) L$ with $\rho(h)$ defined in \eqref{Eq:rho}. }
 \end{lemma}
\begin{proof}
See Appendix~\ref{App:Conv:Off}.
\end{proof}

 Define a positive constant $a''$ as
 \begin{eqnarray}
 &a''=\dfrac{\sigma^2 }{\upsilon}\left\{ 1+  \left[ \frac{L \ln 2}{BT }+W(-e^{-1-\frac{L \ln 2}{BT }})\right] \right.\nn\\ 
&~~~~~~~~~ \left. \times  \exp{\left(\frac{L \ln 2}{BT }+W(-e^{-1-\frac{L \ln 2}{BT } }   )+1\right)} \right\}. \label{Eq:Thresh:Off}
\end{eqnarray}
Then optimizing the objective function of Problem P4 over  the interval $(0, T)$ and investigating the feasibility condition yield the solution  as shown in the following theorem.

 \begin{theorem}\label{Theo:Offload}\emph{
 The optimal offloading duration $t^*$ that solves Problem P4 satisfies the following.
\begin{enumerate}
\item If $P_b h^2< a''$, the problem is infeasible.
\item If $P_b h^2 \ge a''$, $t^*=\rho(h) L$ with $\rho(h)$ defined in \eqref{Eq:rho}.
\end{enumerate}}
 \end{theorem}
  \begin{proof}
See Appendix~\ref{App:Offload}.
\end{proof}
\begin{remark}[Maximum energy savings]\label{Rem:off:egy} \emph{It follows from Theorem~\ref{Theo:Offload}  that the maximum mobile energy savings is
\begin{equation}
S^*_{\textrm{off}} = \upsilon P_b hT - y(h) L\nn  \label{maximum energy savings}
\end{equation}
where the function $y(h)$ is defined as 
\begin{align}
y(h) = \frac{\sigma^2 \ln2}{B h}  \exp \left(W\left(\frac{\upsilon P_b h^2}{\sigma ^2 e} -\frac{1}{e}\right)+1\right). \label{Eq:V} 
\end{align}}
 \end{remark}
\begin{remark} \emph{The tradeoff between the deadline $T$ and BS transmission power $P_b$ as discussed in Remark~\ref{Rem:Trd:Loc} for  the local computing also holds for the current operation mode. Moreover, increasing the channel bandwidth $B$ allows  a more stringent deadline or smaller $P_b$. }  
 \end{remark}
\begin{remark}[BS transmission power vs. offloading duration]\emph{It can be observed from the expression of $t^*$ that increasing the BS transmission power $P_b$ reduces the optimal offloading duration $t^*$. The reason is that higher transmission power is affordable leading to a shorter transmission duration given fixed data to be offloaded. }
\end{remark}
\begin{remark}[Power beacon based MPT]\emph{{\color{black}{A power beacon can be deployed for performing MPT such that the power transfer and offloading will be served by the power beacon and BS, respectively. This will lead to different gains for the MPT and offloading channels, instead of being identical in the current model, which, however, will not cause significant changes to the key results and policy structures.}
 }}
  \end{remark}

\subsection{Offload  or Not?} \label{Section:ModeSel}
Since the mobile has non-causal CSI, for each channel realization, it can decide the operation mode based on whether the 
successful computing conditions are satisfied and which mode achieves the larger mobile energy savings.

First, if only one operation mode is feasible, i.e., $h \ge a/P_b $ for local computing or $h\ge \sqrt{a''/P_b}$ for offloading, then this mode is preferred. 

Next, if both operation modes are feasible, the desirable mode is selected by comparing the amounts of their maximum energy savings. Define the difference between their maximum energy savings as $\Delta S =  S^*_{\textrm{off}} - \bar S^*_{\textrm{loc}} $. It follows from Corollary~\ref{Cor:Egy:Local} and Remark~\ref{Rem:off:egy} that 
\begin{equation}
\Delta S =    \frac{\gamma \theta }{T^2} - y(h) L \label{Eq:S:Diff}
\end{equation}
where the coefficient  $\theta$ satisfies 
\begin{equation}
\left( \sum_{k=1}^{N} p_k^\frac{1}{3}\right)^3 \leq \theta \leq  N^2\sum_{k=1}^{N} p_k  \nn
\end{equation}
and $y(h)$ is given in \eqref{Eq:V}. Then offloading should be performed if and only if  $\Delta S \geq 0$. 

From \eqref{Eq:S:Diff} and \eqref{Eq:V}, the effects of parameters such as the computation deadline $T$ and BS transmission power $P_b$ on the offloading decision are characterized as follows. 

\begin{enumerate}
\item{If $T\le \sqrt{\frac{\gamma \theta }{y(h)L}}$, offloading is preferred which implies that a more strict deadline requirement tends to select the offloading mode. {\color{black}{It can be interpreted as follows. As the deadline increases, the growing rate of harvested energy for local computing is larger than that for offloading. Moreover, the energy consumption for local computing decreases with an increasing deadline, however, that of offloading can be proved to be invariant from \eqref{Eq:V}.
 }}}
\item{If $P_b \le \frac{\sigma^2}{\upsilon h^2}(1+e a''' \ln a''')$ where $a'''=\frac{B h\gamma \theta}{e T^2 \sigma^2 L \ln 2}$, offloading is {\color{black}{selected}} which indicates that offloading is preferred when the BS transmission power is insufficient. The reason for this threshold-based mode selection w.r.t the variation of  BS transmission power is  similar to that of deadline.}
\end{enumerate}

\section{Energy Efficient Mobile Cloud Computing  with a Dynamic Channel}
While wireless channel is assumed fixed in the preceding section, dynamic channel is considered in this section which is modeled as  $M$ channel fading blocks with block duration $T_c$ satisfying $M T_c\!=\!T$. The mobile is assumed to have the prior knowledge of the channel power gains in these fading blocks (see Assumption~\ref{AS:CSI}), enabling off-line data allocation for computing in different fading blocks.
{\color{black}{Similar to the fixed-channel counterpart, the mobile is assumed to select one of the two operation modes for a single task in these fading blocks.}} Local computing and offloading are optimized separately in the following sub-sections. The structures of resultant  control policies are also analyzed. The results can be straightforwardly combined to optimize the  mobile operation mode selection as in Section~\ref{Section:ModeSel} with the details omitted for simplicity. 
\subsection{Energy Efficient Local Computing with a Dynamic Channel}

\subsubsection{Problem Formulation} Exploiting the different time scales for channel variation and computing in Assumption~\ref{As:Channel}, the problem of optimizing  the CPU cycle frequencies for local computing can be decomposed into a master and a slave problems as follows. 
\begin{itemize}
\item {\bf Slave problem:} Given a particular fading block and allocated data, the slave problem aims at minimizing the average energy consumption for computing in this block by controlling the CPU cycle frequencies in a similar way as the fixed-channel counterpart. Consider a single fading block with channel power gain $h$ and allocated data size $\ell$. {\color{black}{Moreover, let  $R \ge 0$ denote the amount of residual energy from computing in the preceding block.}}  The slave problem for this block is formulated as 
\begin{equation}
\begin{aligned}
 \min_{\{f_k\} }  \quad& \sum_{k=1}^{\ell N_0} \gamma p_k f_k^2 \\
\text{s.t.}\quad 
& \sum_{k=1}^m \gamma  f_k^2 \leq R+\upsilon P_b h\sum_{k=1}^m \frac{1}{f_k}, & m = 1,  \cdots, \ell N_0,\\ 
& \sum_{k=1}^{\ell N_0} \frac{1}{f_k}\le T_c,  \\
&  f_k >0,  & k = 1,  \cdots, \ell N_0. 
\end{aligned}
\tag{$\textbf{P5}$} 
\end{equation}
Setting $R = 0$ reduces Problem P5 to P1 for the case of static channel. Let $G_{\textrm{loc}}(\ell, R, h )$ denote the minimum average energy consumption for computing $\ell$-bit data in a single fading block. Then $G_{\textrm{loc}}(\ell, R, h) = \sum_{k=1}^{\ell N_0} \gamma p_k (f_k^*)^2$ where $\{f_k^*\}$ solve the above slave problem. 

\item {\bf Master problem:}  The master problem divides the input data for computing in different fading blocks under the criterion of minimum total energy consumption.  Let $n$ denote the index of fading blocks. The master problem is formulated as follows. 
\begin{equation}
\begin{aligned}
 \min_{\{\ell_n\} }  \ & \sum_{n=1}^{M} G_{\textrm{loc}}(\ell_n, R_n, h_n) \\
\text{s.t.}\
& G_{\textrm{loc}}(\ell_n, R_n, h_n) \!\le\! R_n\!+\!  \upsilon P_b h_n T_c, \!&\!n = 1, \cdots, M,  \\
& R_n=R_{n-1}+\upsilon P_b h_{n-1} T_c \\
&~~~~~~~-G_{\textrm{loc}}(\ell_{n-1}, R_{n-1}, h_{n-1}), & n = 2, \cdots, M,\\
& R_1=0, \qquad \sum_{n=1}^{M} \ell_n = L,\\
&\ell_n \geq 0,  &  n = 1, \cdots, M. 
\end{aligned}
\tag{$\textbf{P6}$} 
\end{equation}

\end{itemize}
{\color{black}{Note that the first two constraints of Problem P6 imply that $R_n \ge 0$ for $n \ge 2$.}}
\subsubsection{CPU-cycle Control  Policy}  The  policy  resulting from solving   Problem P5 can be modified from that obtained from solving Problem P1 for the case of static channel. For ease of notation,  define two constants: 
\begin{equation}\label{Eq:Const:L}
b\!=\!\l(\frac{\upsilon P_b h T_c^3 + R T_c^2}{\gamma \theta_0}\r)^{\frac{1}{3}}, \ b'\!=\!\l(\frac{\upsilon P_b h T_c^3+ R T_c^2}{\gamma \theta_1}\r)^{\frac{1}{3}}
\end{equation}
and two energy consumption functions 
\begin{equation}
\bar{E}_0(\ell)=\frac{\gamma \phi_0 \ell^3}{T_c^2}, \qquad \bar{E}_1(\ell)=\frac{\gamma \phi_1 \ell^3}{T_c^2} \nn
\end{equation}
{\color{black}{where $\ell$ is the input-data size and $\theta_0, \theta_1, \phi_0$ and $\phi_1$ are
 the scaling factors of $\ell^3$ for the two thresholds and energy consumption functions determined by the system model and CCI (see Appendix C of \cite{zhang:MobileMmodel:2013})}}, with $\theta_0 > \theta_1$ and $\phi_0 < \phi_1$.
  Following the same procedure as for deriving Theorem~\ref{Theo:LocalComp} and Corollary~\ref{Cor:Egy:Local}, the optimal policy for the current case is obtained as follows. 
\begin{corollary}\label{Cor:DataThreshold}\emph{Consider an arbitrary channel fading block with the corresponding input-data size  $\ell$ and residual  energy $R$. The optimal CPU-cycle frequencies $\{f_k^*\}$ and the minimum average energy consumption $G_{\textrm{loc}}(\ell, R, h)$ are described as follows.
\begin{enumerate}
\item If $\ell \le b$, 
\begin{equation}
f_k^* \!=\!\l(\frac{1}{T_c}\sum_{m=1}^{\ell N_0} p_m^{\frac{1}{3}} \r) p_k^{-\frac{1}{3}}, ~ \forall k, ~ \text{and}~~G_{\textrm{loc}}(\ell, R, h)=\bar{E}_0(\ell). \nn 
\end{equation}
\item If $b< \ell \le b' $, 
\begin{equation}
f_k^*\!=\!\l[\frac{1}{T_c}\sum_{m=1}^{\ell N_0} (p_m+\lambda)^{\frac{1}{3}} \r] (p_k + \lambda)^{-\frac{1}{3}}, \quad \forall k \nn
\end{equation}
where the positive constant $\lambda$ satisfies 
\begin{equation}
\left[ \sum_{k=1}^{\ell N_0}  \l(p_k+\lambda \r) ^{\frac{1}{3}}\right]^2 \left[ \sum_{k=1}^{\ell N_0} ({p_k+\lambda})^{-\frac{2}{3}}\right]=\frac{\upsilon P_b h T_c^3 +R T_c^2}{\gamma}; \nn
\end{equation}
 and
\begin{equation}
\bar{E}_0(\ell)<G_{\textrm{loc}}(\ell, R, h) \le \bar{E}_1(\ell). \nn
\end{equation}
Moreover, when $\ell  \to b$, $G_{\textrm{loc}}(\ell, R, h) \to \bar{E}_0(\ell)$ and when  $\ell \to b'$, $G_{\textrm{loc}}(\ell, R, h) \to \bar{E}_1(\ell)$.
\item If $ \ell > b'$, $\{f^*_1, f^*_2, \cdots, f^*_{\ell N_0}\}$ is an empty set. 
\end{enumerate}}
\end{corollary}

This Corollary shows that the slave problem is feasible only if $\ell \le b'$ where $b'$ is determined by the channel power gain $h$ and the residual energy $R$. Moreover, when $\ell \le b$, the CPU-cycle frequencies and energy consumption are independent of $P_b h$,  implying  that the energy harvesting constraint is inactive for small data-input size. However, for the case of $b< \ell \le b'$, both $\{f_k^*\}$ and $G_{\textrm{loc}}(\ell, R, h)$ are determined by $P_b h$, which is consistent with the case of static channel.

\subsubsection{Sub-optimal Data Allocation Policy}\label{Subsubs:SODAP} The derivation for the optimal data allocation policy is intractable due to the lack of a closed-form expression for the energy consumption function, $G_{\textrm{loc}}(\ell, R, h)$, which can be observed from the solution for the slave problem. In this sub-section, a sub-optimal but simple policy is derived by the approximation for $G_{\textrm{loc}}(\ell, R, h)$ and the amounts of residual energy in all  fading blocks, denoted as $\{R_n\}$.


First, the proposed sub-optimal data allocation policy requires only the following properties of $G_{\textrm{loc}}(\ell, R, h)$. 
\begin{assumption}\label{Ass:ConvexityG}\emph{
For an arbitrary block fading channel with data-input size $\ell$, $G_{\textrm{loc}}(\ell, R, h)$ is a monotone-increasing, differentiable and convex function for  $ \ell \in [0,   b']$.}
\end{assumption}
The assumption can be justified as follows. First, the monotonicity of $G_{\textrm{loc}}(\ell, R, h)$ arises from the fact that computing more data requires additional energy.  Next, finite computing energy per bit gives the rationality for the differentiability property. Last, the second derivative of $G_{\textrm{loc}}(\ell, R, h)$ relates to the increase in  energy consumption per additional data bit. For the special case of equal CPU-cycle frequencies,  the energy consumption of $ \frac{ \gamma N_0^3 \ell^2}{T_c^2}$ per bit grows with data-input size $\ell$, supporting the assumption on the  convexity of $G_{\textrm{loc}}(\ell, R, h)$.

Next, the residual energy variables $\{R_n\}$ are  approximated. To this end, $R_n$ can be bounded as follows. 
\begin{lemma}\label{Lem:PositiveResidual}\emph{
Given that local computing of input data is feasible, then $R_1 = 0$ and 
\begin{equation}
\bar{\phi}(\upsilon P_b h_{n-1} T_c + R_{n-1}) \leq R_n \leq \upsilon P_b h_{n-1} T_c + R_{n-1}
\end{equation}
for $n =   2, \cdots, M$ where $\bar{\phi}=1-\frac{\phi_1}{\theta_1} \ge 0$ .}
\end{lemma}
\begin{proof}
See  Appendix~\ref{App:PositiveResidual}.
\end{proof}

Note that the lower bound on $R_n$ is nonzero due to the energy harvesting constraint.
{\color{black}{Since it is difficult to obtain $R_n$ in closed form for the same reason as for deriving $G_{\textrm{loc}}(\ell, R, h)$  and the upper bound corresponds to the case  without  data computing, $R_n$ is replaced by its lower bound in the design of the proposed sub-optimal data allocation policy and the result is represented by $\hat{R}_n$.}} In other words,  $\hat{R}_n =\bar{\phi}
(\upsilon P_b h_{n-1} T_c + \hat{R}_{n-1} )$ for $n \ge2$ and $\hat{R}_1 \!=\! 0$. Correspondingly, the constants $b$ and $b'$ defined in \eqref{Eq:Const:L} are modified as 
\begin{equation*}
\hat{b}_n\!=\!\l(\frac{\upsilon P_b h_n T_c^3 \!+\! \hat{R}_n T_c^2 }{\gamma \theta_0}\r)^{\frac{1}{3}}\! \text{and}~~ \hat{b}'_n\!=\!\l(\frac{\upsilon P_b h_n T_c^3 \!+\! \hat{R}_n T_c^2 }{\gamma \theta_1}\r)^{\frac{1}{3}}. 
\end{equation*}
Using the above approximation and definitions, the minimum average energy consumption of the $n$-th fading block, denoted as $\hat{G}_{\textrm{loc}}$, follows from Corollary~\ref{Cor:DataThreshold} and Assumption~\ref{Lem:PositiveResidual} as
\begin{equation}
    \hat{G}_{\textrm{loc}}(\ell_n, \hat{R}_n, h_n)=
   \begin{cases}
   \frac{\gamma \phi_0 \ell_n^3}{T_c^2}, &\mbox{if $ \ell_n \le \hat{b}_n$}\\
   g( \ell_n), &\mbox{if $\hat{b}_n< \ell_n \le \hat{b}'_n $}
   \end{cases} \label{Eq:UniformG}
\end{equation}
where  $g( \ell_n)$ is a general function such that   $\hat{G}_{\textrm{loc}}(\ell_n, \hat{R}_n, h_n)$ has the properties in Assumption~\ref{Ass:ConvexityG}. 

Based on the above approximations,  the data-allocation problem for minimizing the total energy consumption can be readily reformulated in a simple form as follows. 
{\center{\textrm{(Sub-optimal Data Allocation)}}}
\begin{equation}
\begin{aligned}
\underaccent{\{\ell_n\} } \min  \qquad  &\sum_{n=1}^M \hat{G}_{\textrm{loc}}(\ell_n, \hat{R}_n, h_n)\\
\text{s.t.}\qquad 
& \sum_{n=1}^M \ell_n = L,\\ 
& 0 \le \ell_n \le \hat{b}_n' , & n = 1, 2, \cdots, M.  
\end{aligned}
\tag{$\textbf{P7}$} 
\end{equation}
Problem P7 is a convex optimization problem. The structure of the resultant data allocation policy can be characterized as follows. 
Let  $b_n(\xi)$ denote  the root of equation: $\frac{\partial{\hat{G}_{\textrm{loc}}}}{\partial \ell_n}(b_n, \hat{R}_n, h_n)=\xi $ where $\xi$ is a Lagrange multiplier. Then the main result of this sub-section is obtained as shown below. 

\begin{proposition}\label{Pro:Data:Local}\emph{If $ L \le \sum_{n=1}^M \hat{b}'_n $, Problem P7 is feasible.
And the data-allocation policy that solves Problem P7 is given as 
\begin{equation}
    \ell_n^*=
   \begin{cases}
   \hat{b}'_n,  &\mbox{ $h_n \le\frac{\gamma \theta_1 b_n^3(\xi)-\hat{R}_n T_c^2}{\upsilon P_b T_c^3}$}\\
   b_n(\xi), &\mbox{ $h_n >\frac{\gamma \theta_1 b_n^3(\xi)-\hat{R}_n T_c^2}{\upsilon P_b T_c^3}.$} 
   \end{cases} \label{Eq:SubDataLocal}
\end{equation}
}\end{proposition}
\begin{proof}
See Appendix~\ref{App:Data:Local}.
\end{proof}

Note that $\!\{\ell_n^*\} \! $ are nonzero, indicating that the policy utilizes \emph{all} fading blocks for computing since local computing and MPT can be performed simultaneously over all fading blocks.

\subsection{Energy Efficient Offloading with a Dynamic Channel}
\subsubsection{Problem Formulation}  Following the local computing counterpart, the problem for optimal computation offloading is formulated using the master-and-slave model as follows.

\begin{itemize}
\item {\bf Slave problem:} Given fixed allocated data, residual energy  and channel power gain, the slave problem aims at finding the optimal time division of a fading block for separate energy harvesting and offloading. Consider a single fading block with channel power gain $h$, allocated data-input size $\ell$ and the residual energy $R$ which comes from offloading in the preceding block. The slave problem is formulated as follows for maximizing the energy savings in this block with the optimal time division. 
\begin{equation}
   \begin{aligned}
\max_t \qquad  & E_{\textrm{MPT}}(t, h)-E_{\textrm{off}}(t, h) \\  
\text{s.t.} \qquad 
& 0< t < T_c,\\
& R+ E_{\textrm{MPT}}(t, h)-E_{\textrm{off}}(t, h)\ge 0
\end{aligned}
\tag{$\textbf{P8}$} 
\end{equation}
where $E_{\textrm{MPT}}(t, h)\!=\!\upsilon P_b h (T_c - t)$ and $E_{\textrm{off}}(t, h)\!=\!(2^ \frac{\ell}{Bt}-1)\frac{\sigma^2}{h} t$.
 Note that setting $R\!=\! 0$ reduces Problem P8 to P4. Let  $G_{\textrm{off}}(\ell, R, h)$ denote  the maximum energy savings for offloading the $\ell$-bit data in this single fading block. Then $G_{\textrm{off}}(\ell, R, h)=E_{\textrm{MPT}}(t^*, h)-E_{\textrm{off}}(t^*, h)$ where $t^*$ solves the above slave problem.

\item {\bf Master problem:} The master problem concerns the optimal data allocation for offloading in different fading blocks with the objective of maximizing the  total energy savings. Let $n$ denote the index of fading blocks. Given the solution for the slave problem, the master problem of the optimal data allocation over different fading blocks is formulated as follows. 
\begin{equation}
\begin{aligned}
 \max_{\{\ell_n\} }  \quad  & \sum_{n=1}^{M} G_{\textrm{off}}(\ell_n,  R_n, h_n) \\
\text{s.t.}\quad
& R_n = \sum_{m=1}^{n-1} G_{\textrm{off}}(\ell_m,  R_m, h_m) , &n = 2, \cdots, M,\\
& R_n \geq 0 , & n = 2, \cdots, M,\\
& R_1 = 0,  \qquad  \sum_{n=1}^{M} \ell_n = L,\\
& \ell_n \ge 0, & n = 1, 2, \cdots, M. \nn
\end{aligned}
\tag{$\textbf{P9}$} 
\end{equation}

\end{itemize}
\subsubsection{Optimal Time Division Policy} The slave problem, Problem P8,   can be modified from the fixed-channel counterpart, Problem  P5,  by adding the residual energy and thus solved following a similar procedure. For this purpose, define the following constants:
\begin{equation*}
 c \!=\! \frac{T_c B\left[1+ W(\frac{\upsilon P_b h^2}{\sigma ^2 e} -\frac{1}{e})\right]}{\ln 2} ~\text{and} ~~ c' \!=\! B T_c\log \l(1\!+\! \frac{R h}{\sigma^2 T_c}\r).
\end{equation*}
Then the optimal time division policy for the current case with residual energy $R$ is obtained as shown  in  Corollary~\ref{Cor:EnergyThreePoli}, following a similar procedure as for deriving Theorem~\ref{Theo:Offload}. 

\begin{corollary}\label{Cor:EnergyThreePoli}\emph{
Consider an arbitrary channel fading block with $\ell$-bit data input and residual energy $R$. The optimal offloading duration $t^*$ and the maximum energy savings $G_{\textrm{off}}(\ell, R, h)$ in this block are given as follows.
\begin{enumerate}
\item If either: (a) $R \le \frac{B T_c y(h) }{\ln 2}-\frac{\sigma^2}{h} T_c$ and $\ell \le \frac{\upsilon P_b h T_c+R}{y(h)}$ or \\
\quad (b)  $R > \frac{B T_c y(h)}{\ln 2}-\frac{\sigma^2}{h} T_c$ and $\ell < c$, then
\begin{equation}
    t^*=\rho(h) \ell  \quad \text{and} \quad G_{\textrm{off}}(\ell, R, h)=\upsilon P_b h T_c - y(h) \ell \nn
\end{equation}
where the constants $y(h)$ and $\rho(h)$ are defined  in \eqref{Eq:V} and \eqref{Eq:rho}, respectively.\\
\item If $R > \frac{B T_c y(h)}{\ln 2}-\frac{\sigma^2}{h} T_c$ and $c\le \ell \le c'$, then
\begin{equation}
    t^*=T_c \quad \text{and} \quad G_{\textrm{off}}(\ell, R, h)=-(2^ \frac{\ell}{B T_c}-1)\frac{\sigma^2}{h} T_c . \nn
\end{equation}
\item For other combinations of conditions for $R$ and $\ell$, Problem P8 is infeasible.
\end{enumerate}}
\end{corollary}

In the above Corollary, the largest data-input size, optimal time division and maximum energy savings are determined by the channel power gain $h$ and the residual energy $R$. Specifically, the time division of case 1) is the same as that of  static channel case and the maximum energy savings decreases {\color{black}{linearly}} with the growing of data-input size.  Case 2) corresponds to the scenario where the mobile has large residual energy and data input size. In this case, spending all time on offloading is the optimal time division policy.

\subsubsection{Sub-optimal Data Allocation Policy}  One can observe from problem P9 that different summation terms in the objective function are coupled due to  the residual energy delivered from one block to the next. The conventional approach for solving this type of optimization problem is using dynamic programming (DP). The state space for the resultant DP is continuous and has to be discretized to facilitate iterative computation for the optimal policy, bringing high complexity. However, simulation reveals the sub-optimal low-complexity policy  to be presented shortly can achieve close performance as  DP. More importantly, the DP approach yields little insight to the structure of optimal policy while the said sub-optimal policy allows data allocation to be derived in closed form.  

The proposed sub-optimal policy is obtained by setting the residual energy variables as zero: $\{R_n\} = \{0\}$, which is observed from the energy harvesting constraint in Problem P9 to be their lower bounds. Combining the approximation of $\{R_n\}$ with Corollary~\ref{Cor:EnergyThreePoli} reduces Problem P9 as 
{\center{\textrm{(Sub-optimal Data Allocation)}}}
\begin{equation}
\begin{aligned}
 \max_{\{\ell_n \}}  \quad  &\sum_{n=1}^M \upsilon P_b h_n T_c - y(h_n) \ell_n\\
\text{s.t.}\quad 
& \sum_{n=1}^M \ell_n = L,\\ 
& 0 \le \ell_n \le \frac{\upsilon P_b h_n T_c}{y(h_n)}, & n = 1, 2, \cdots, M.
\end{aligned}
\tag{$\textbf{P10}$} 
\end{equation}
Problem P10 is a convex optimization problem and solving it gives the optimal policy in closed form. To state the policy, let the channel power gains $\{h_n\}$ be rearranged and re-denoted as $\{\tilde{h}_n\}$  such that  $\{y(\tilde{h}_n)\}$ are in ascending order: $y(\tilde{h}_1) \leq y(\tilde{h}_2) \le  \cdots \le y(\tilde{h}_M)$. Moreover, let $\Pi$ represent the permutation matrix with 
\begin{equation}
[\tilde{h}_1, \tilde{h}_2, \cdots, \tilde{h}_M]^T = \Pi \times [h_1, h_2, \cdots, h_M]^T.  \nn
\end{equation}

Using these definitions, the optimal policy from solving Problem P10 is given in the following proposition. 

\begin{proposition} \label{Pro:Sub:Off}\emph{If $L \le \sum_{n=1}^M  \frac{\upsilon P_b h_n T_c}{y(h_n)}$, Problem P10 is feasible. And it can be observed that the data-allocation policy solving Problem P10 is:
\begin{equation}
[\ell_1^*,  \ell_2^*, \cdots, \ell_M^*]^{T} ={\Pi}^{-1} \times [\tilde{\ell}_1^*,  \tilde{\ell}_2^*, \cdots, \tilde{\ell}_M^*]^{T}  \nn
\end{equation}
with $\{\tilde{\ell}_n^* \}$ given in the following
\begin{equation}
    \tilde{\ell}_n^* =
   \begin{cases}
   \frac{\upsilon P_b \tilde{h}_n T_c}{y(\tilde{h}_n)}, &\mbox{ $n=1, 2, \cdots j$}\\
   L- \sum_{k=1}^j \frac{\upsilon P_b \tilde{h}_k T_c}{y(\tilde{h}_k)} &\mbox{ $n=j+1$},\\
   0, &\mbox{otherwise}
   \end{cases}\label{Equ:SubDataOff}
\end{equation}
where $j$ is unique and satisfies:
\begin{equation}
\sum_{n=1}^j \frac{\upsilon P_b \tilde{h}_n T_c}{y(\tilde{h}_n)} < L\le \sum_{n=1}^{j+1} \frac{\upsilon P_b \tilde{h}_n T_c}{y(\tilde{h}_n)}. \nn
\end{equation}
}
\end{proposition}

The above data allocation policy is a greedy approach which allocates data to the fading blocks sequently by the ascending order of $y(h_n)$ 
until all the  input data  has been allocated.

\section{Simulation Results}
In this section, the performance of wirelessly powered mobile cloud computing   with  static and dynamic channels is investigated by simulation. The parameters are set as follows unless specific otherwise. The data input size $L$ is $1000$-bit and the number of CPU cycles required for per bit is modeled by a Gamma distribution with $\alpha =4$ and $\beta =200$ as in \cite{zhang:MobileMmodel:2013}, resulting in $\epsilon=0.05$. The constant determined by the switch capacitance $\gamma$ is $10^{-28}$.
 The energy conversion efficiency $\upsilon$ is $0.8${\color{black}\cite{Sun:EgyConv:2012}}. Let $N_t$ denote the number of BS antennas and set as $N_t = 2$. The $N_t\times 1$  vector channel, denoted as $\bh$, is assumed to follow Rician fading and thus modeled as 
\begin{equation}
\mathbf{h}=\sqrt{\frac{\Omega K}{1+K}} \mathbf{\bar{h}}+\sqrt{\frac{\Omega}{1+K}}\mathbf{h_w} \nn
\end{equation}
where the Rician factor $K \in \{0, 10\}$, the average fading power gain $\Omega =5\times 10^{-6}$, line-of-sight (LoS) component $\mathbf{\bar{h}}$ has all elements equal to one, and $\mathbf{h_w}$ is a $N_t\times 1$ i.i.d. $\mathcal{CN}(0, 1)$ vector representing small-scale fading. The effective channel power gain $h=\l \| \mathbf{h} \r \|^2$, resulting from the transmit/receive beamforming. In addition, the variance of complex white Gaussian channel noise $\sigma^2$ is $10^{-9}$ W  and the channel bandwidth $B$ is $1$ MHz.

\subsection{Static Channel}

The performance of three polices is evaluated, including the optimal local computing, optimal offloading and optimal mobile mode selection (MMS) integrating the preceding two polices. The baseline schemes for optimal local computing and optimal offloading are local computing with equal CPU-cycle frequencies and offloading with equal time partition (for MPT and offloading), respectively. 

\begin{figure}[t!]
\centering
\subfigure[Rician factor $K=0$]{\label{fig}
\includegraphics[width=8.5cm]{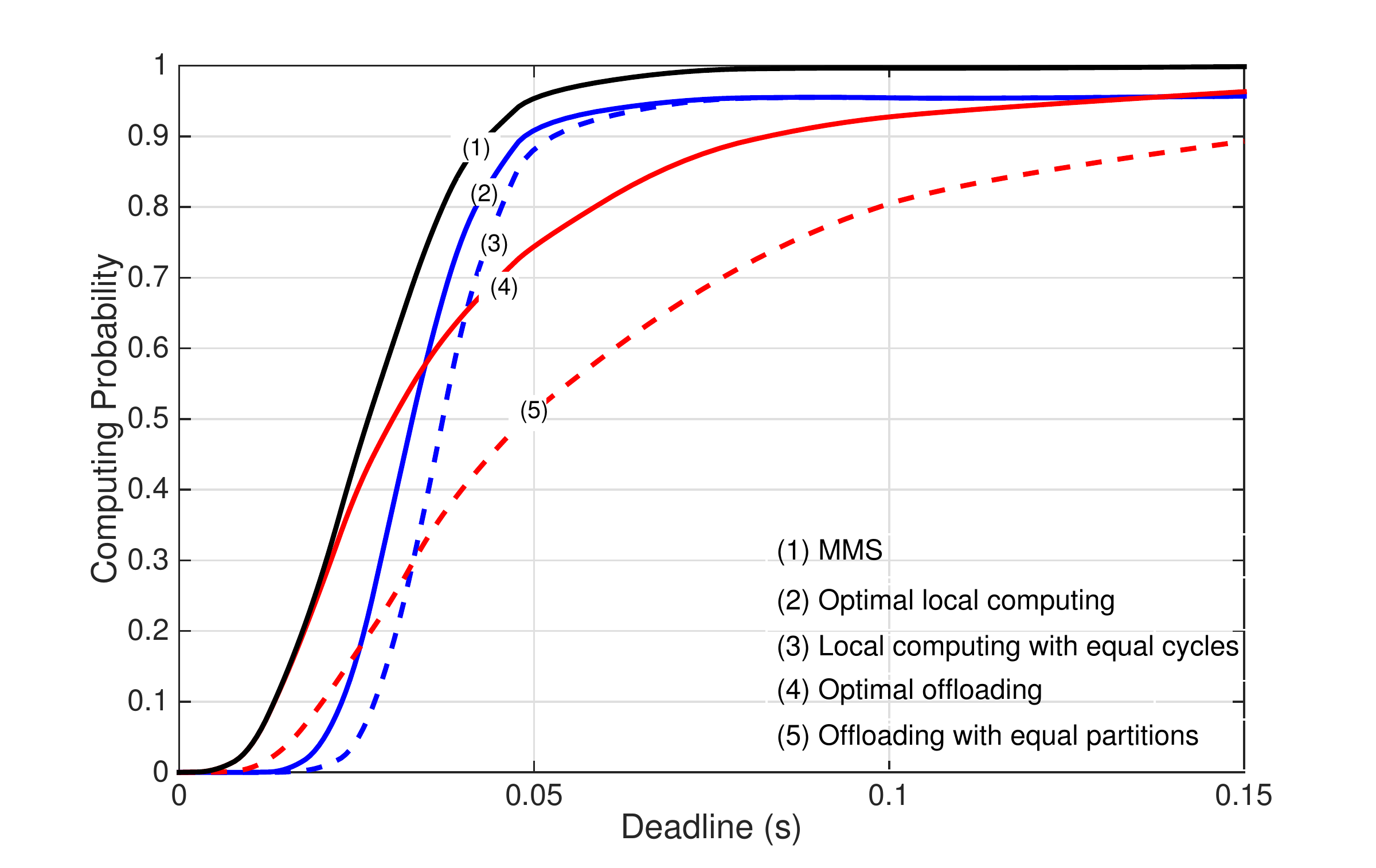}}
\subfigure[Rician factor $K=10$]{\label{fig}
\includegraphics[width=8.5cm]{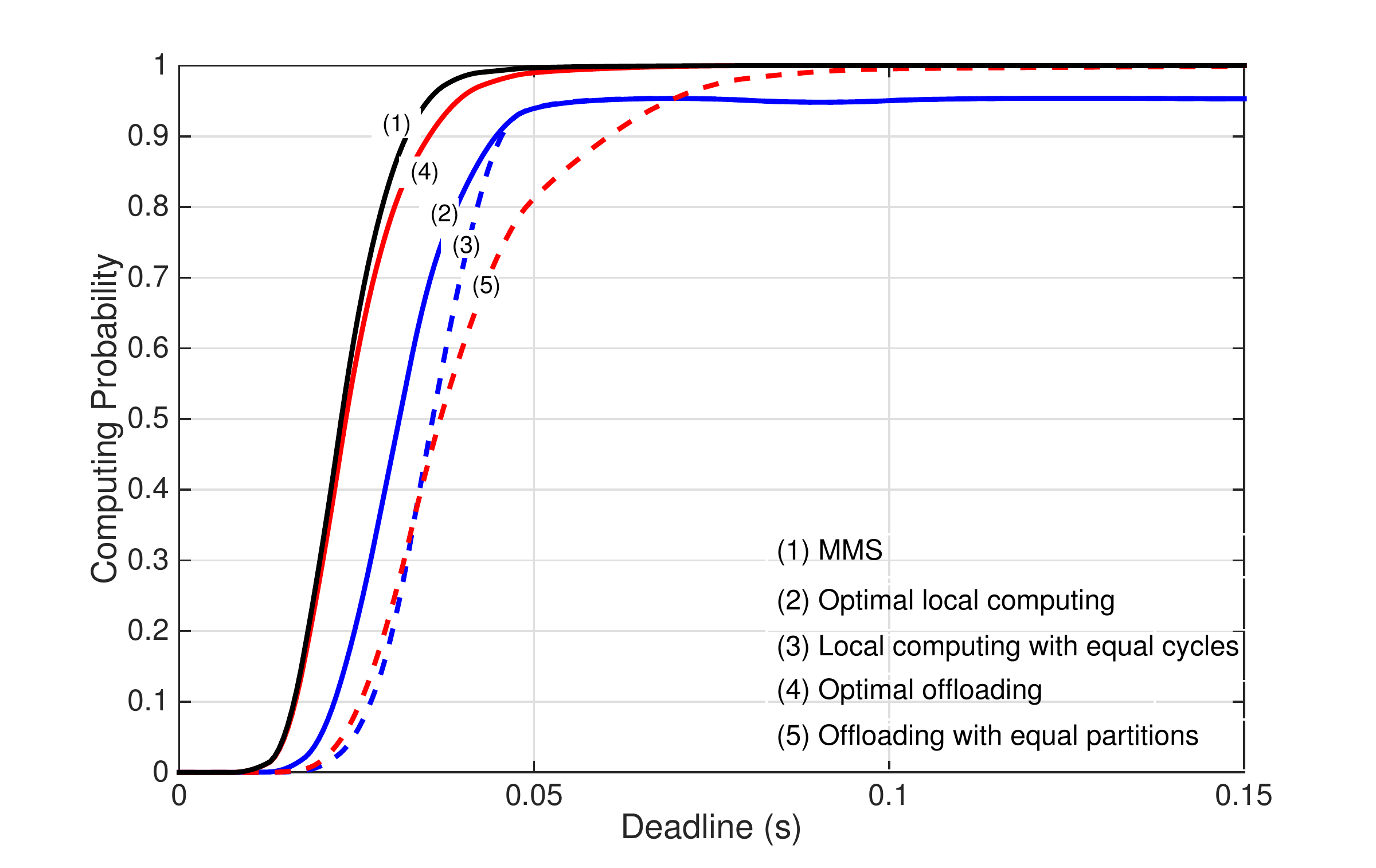}}
\caption{Effect of deadline on the computing probability for the case of static channel. The BS transmission power is fixed at $P_b \!=\! 0.5$ W. }
\label{Fig:Sim:Static}
\end{figure}
\begin{figure}[t!]
\centering
\subfigure[Rician factor $K=0$]{\label{fig}
\includegraphics[width=8.5cm]{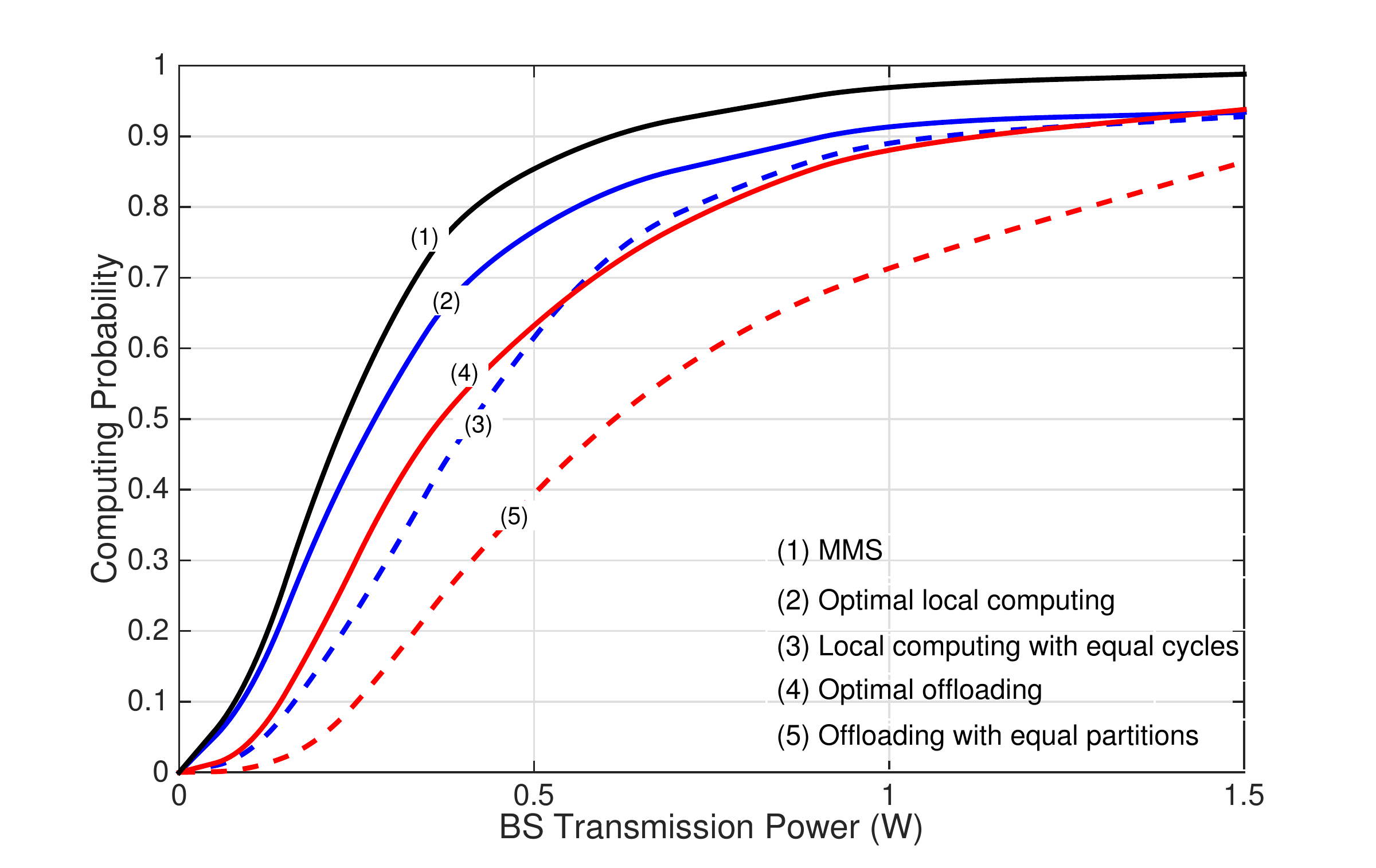}}
\subfigure[Rician factor $K=10$]{\label{fig}
\includegraphics[width=8.5cm]{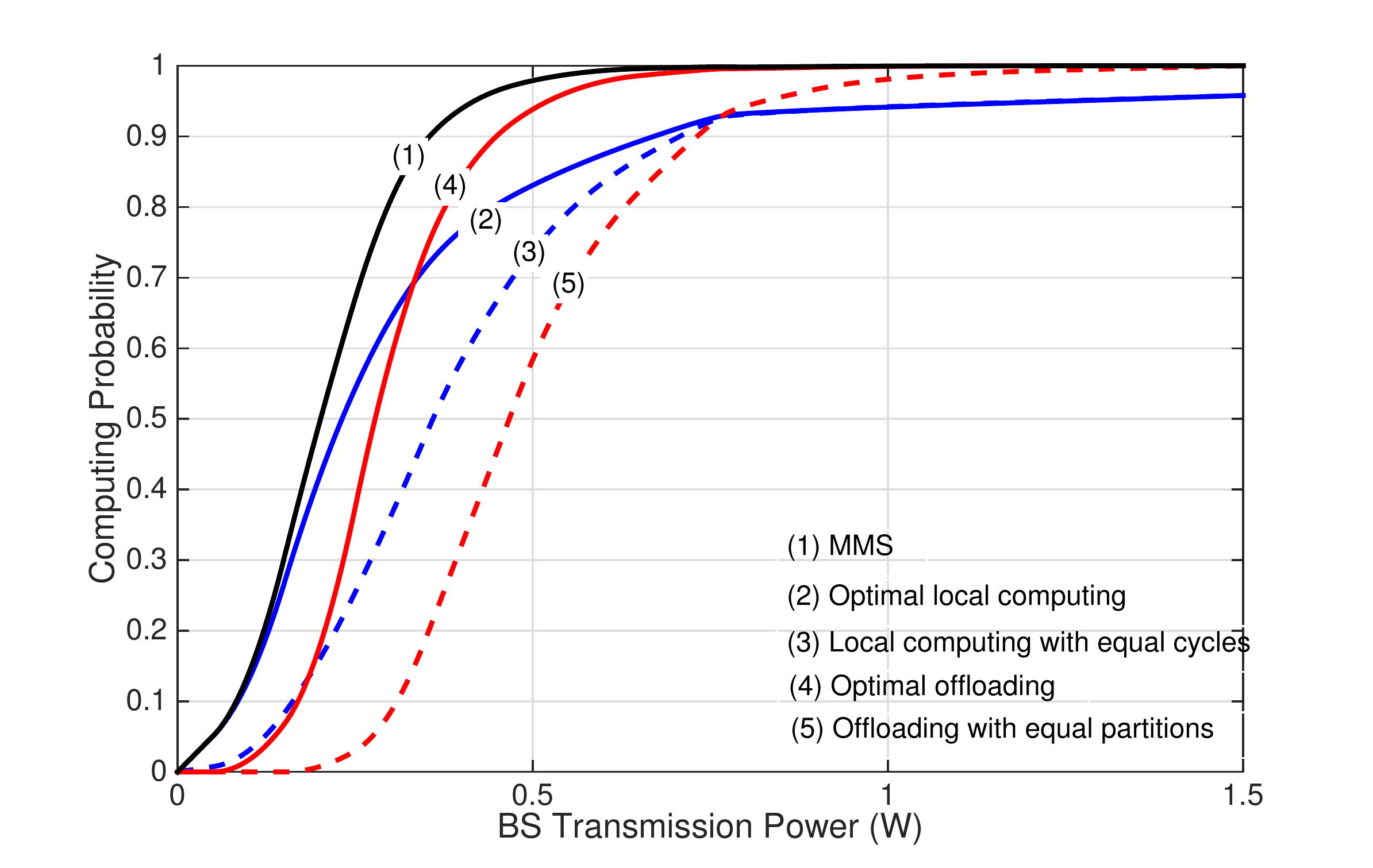}}
\caption{Effect of BS transmission power on the computing probability for the  static channel. The  deadline is fixed at $T \!\!=\! 0.035$ s.}
\label{Fig:Sim:Static:a}
\end{figure}

\begin{figure}[h!]
\begin{center}
\includegraphics[width=9cm]{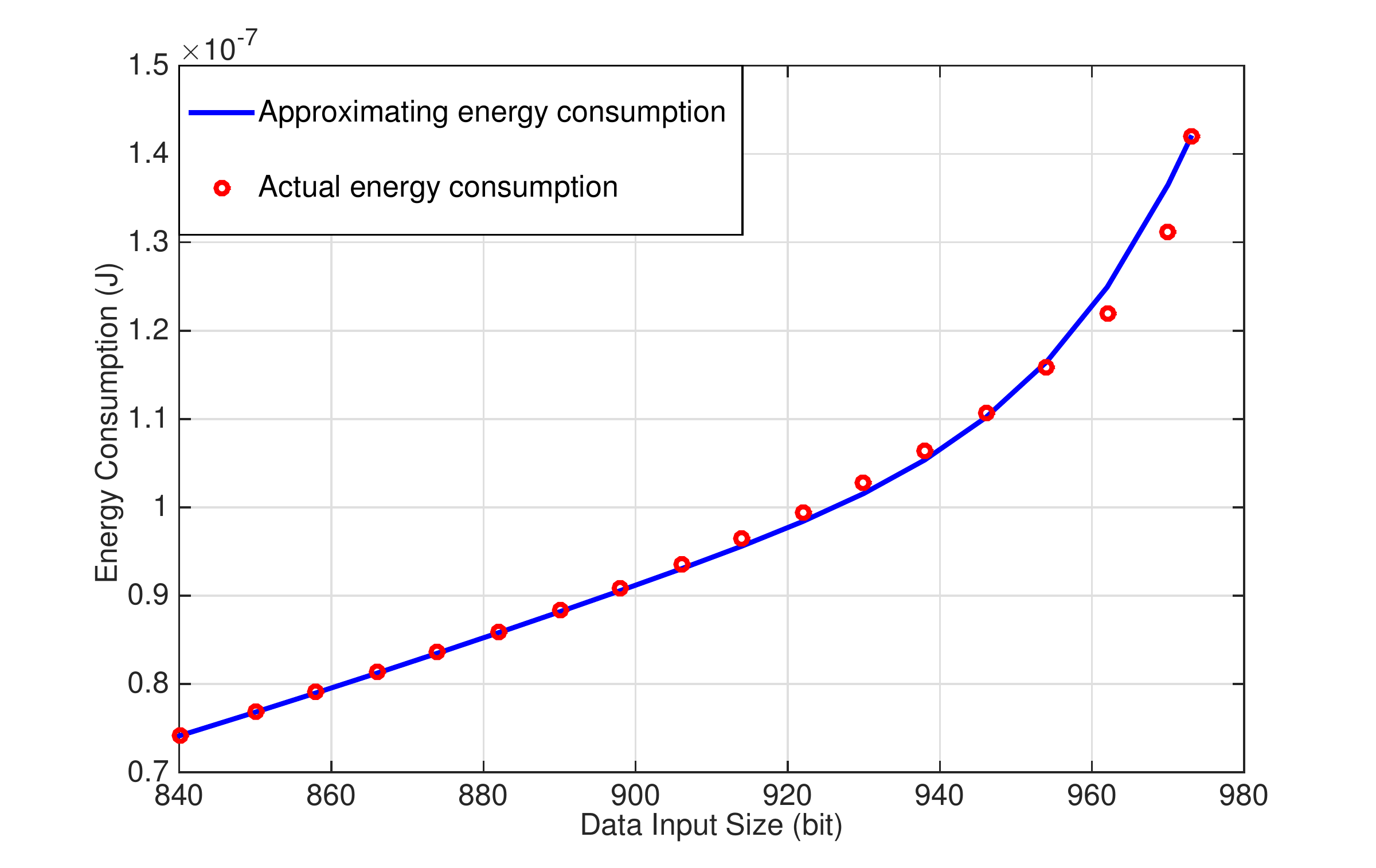}
\caption{Approximation $\hat{G}_{\textrm{loc}}$ with $P_b= 1$ W and $T=0.035 s$.}
\label{Fig:Approximation performance}
\end{center}
\end{figure}
Fig.~\ref{Fig:Sim:Static} (a) and (b) show the curves of computing probability versus deadline $T$ for the Rician factor $K= 0$ and $10$, respectively. Several observations can be made.  First, computing probability is observed to be a monotone-increasing function of   $T$. {\color{black}{Next, for a highly random channel ($K = 0$), the crossing of the curves (2) and (4) suggests mode switching as the deadline increases. Specifically, local computing and offloading are preferred for relatively loose and strict deadlines, respectively. The reason is that compared with offloading, the computing probability for local computing grows faster as the deadline increases and also decays faster as the deadline decreases. Note that the thresholds for mode switching w.r.t the computing probability has no simple closed form.}} 
However, given LoS ($K=10$), the optimal offloading is always preferred since the required transmission energy is small for such a channel. Last, compared with their corresponding baseline schemes, optimizing offloading shows more substantial performance gain than optimizing local computing.

The curves of computing probability versus BS transmission power $P_b$ are plotted in Fig.~\ref{Fig:Sim:Static:a} (a) for $K=0$ and in Fig.~\ref{Fig:Sim:Static:a} (b) for $K=10$ with a fixed computing deadline $T \!=\! 0.035$s. As observed from the figures, for a highly random channel ($K = 0$),  the optimal local computing is preferred to offloading. However, for a channel with LoS ($K = 10$), the optimal local computing is preferred only when $P_b$ is small while optimal offloading should be chosen for large $P_b$. Other observations are similar to those from Fig.~\ref{Fig:Sim:Static}.  


\begin{figure}[t!]
\centering
\subfigure[Rician factor $K=0$]{\label{fig}
\includegraphics[width=9cm]{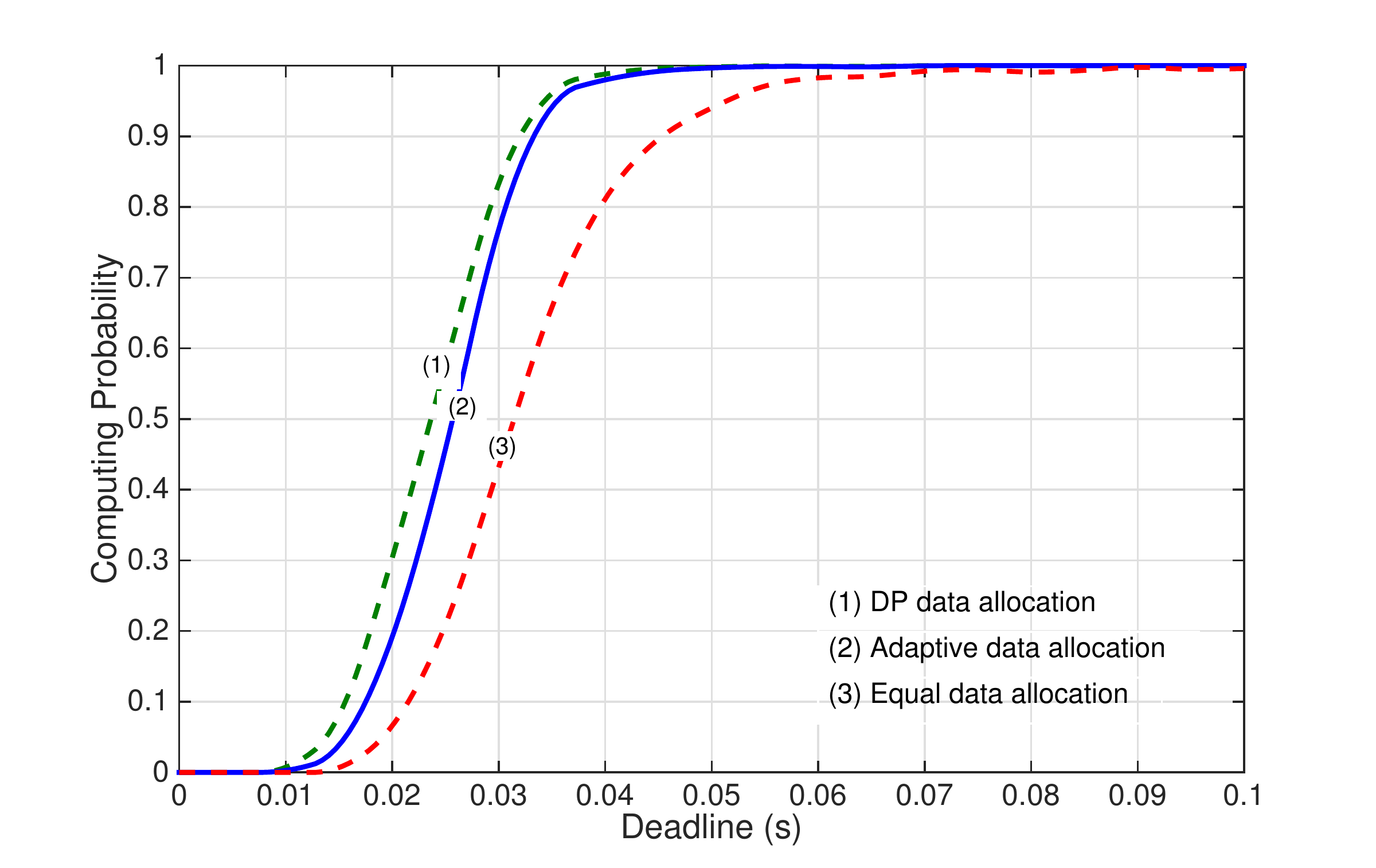}}
\subfigure[Rician factor $K=10$]{\label{fig}
\includegraphics[width=9cm]{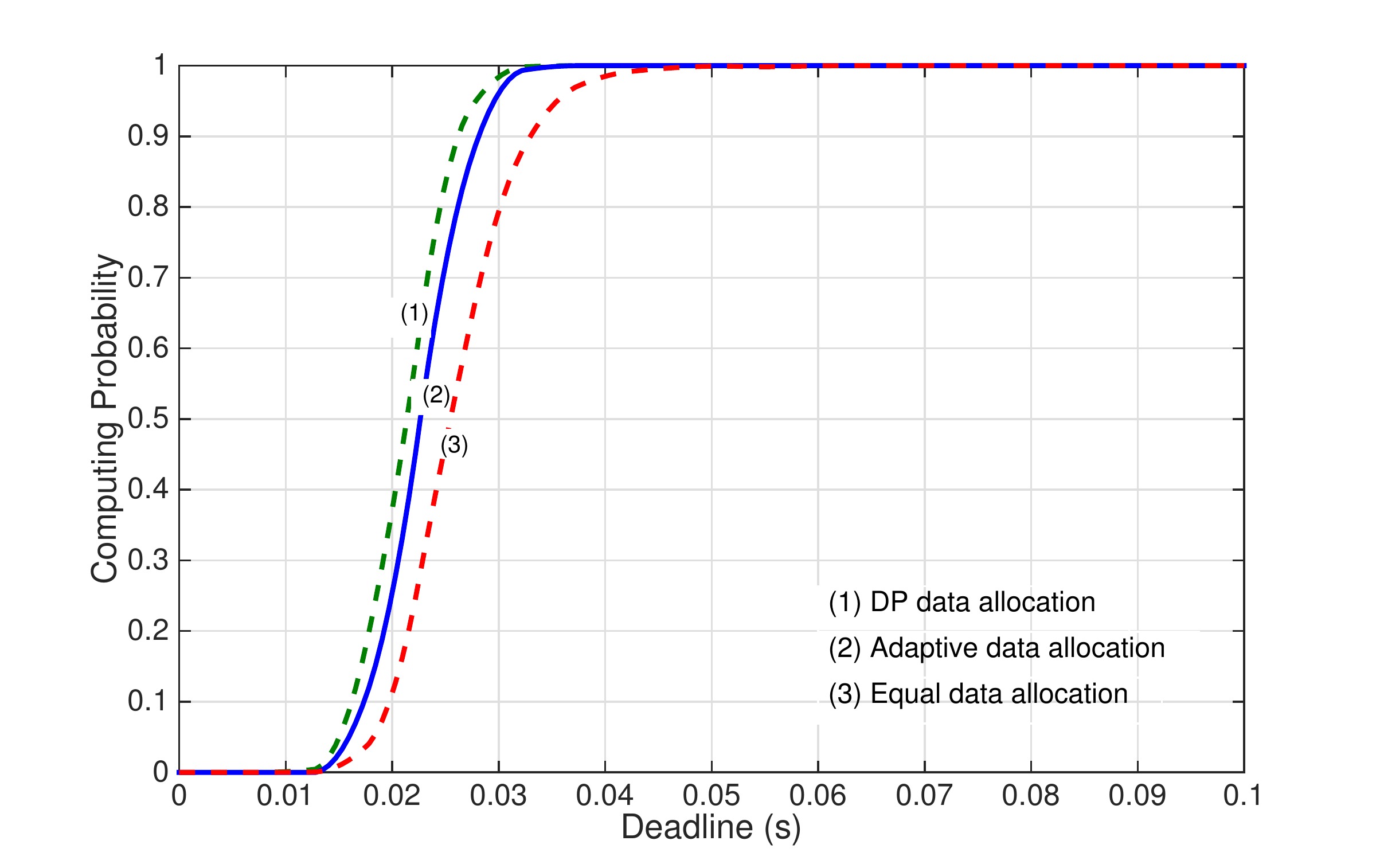}}
\caption{Effect of deadline on the computing probability for the  dynamic  channel. The BS transmission power is fixed at $P_b \!=\! 0.5$ W. }
\label{Fig:Sim:DyChan}
\end{figure}

\begin{figure}[t!]
\centering
\subfigure[Rician factor $K=0$]{\label{fig}
\includegraphics[width=9cm]{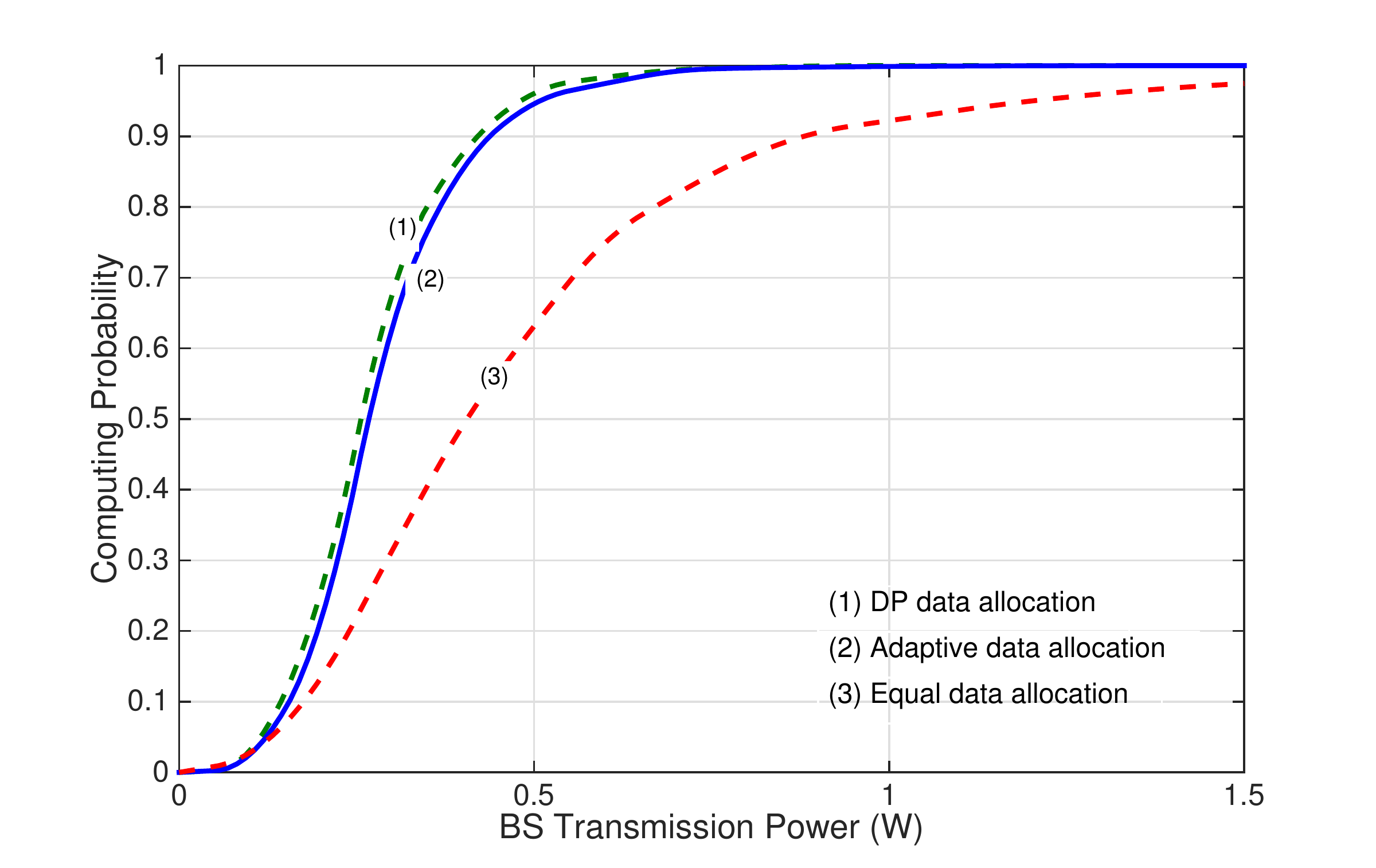}}
\subfigure[Rician factor $K=10$]{\label{fig}
\includegraphics[width=9cm]{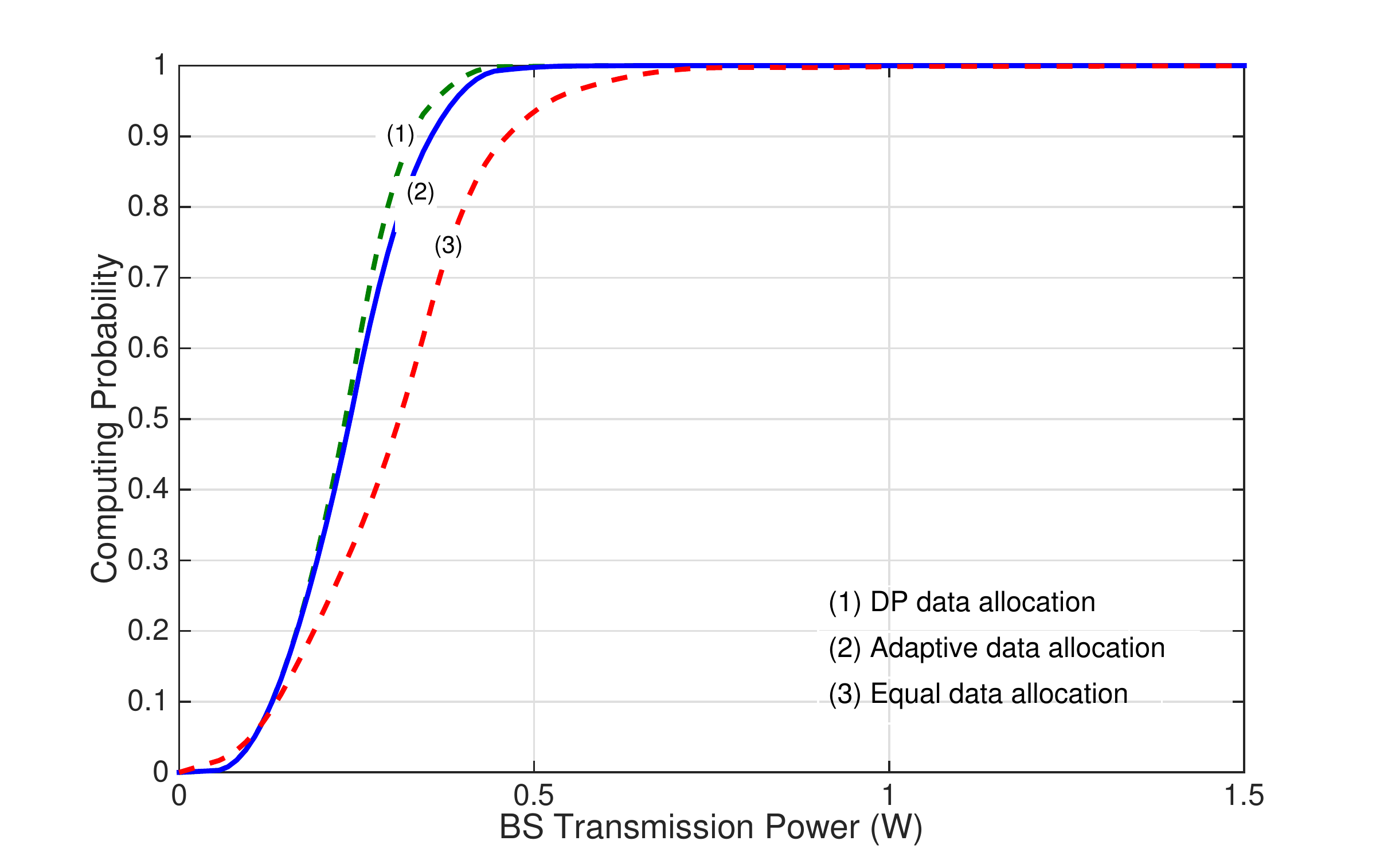}}
\caption{Effect of BS transmission power on the computing probability for the dynamic channel. The computing  deadline is fixed at $T = 0.035$ s.}
\label{Fig:Sim:DyChan:a}
\end{figure}

Consider the case of dynamic channel and data allocation over $4$ fading blocks. {\color{black}{The relatively small number of fading blocks is assumed to account for the difficulties of large-range channel prediction in practice.}}
The sub-optimal data allocation policies derived in the preceding section are compared with two baseline polices: the optimal one computed based on the derived sub-optimal data allocation policy for local computing and DP policy for offloading, referred  as \emph{DP data allocation}, as well as a simple policy based on \emph{equal data allocation}.  {\color{black}{Specifically, the DP policy for offloading is obtained by discretizing the state space and solving the Bellman equation backward recursively \cite{Bertsekas:DP:1995}. Each policy is feasible if one of its two operation modes is possible. 
}}For local computing, the  minimum average energy consumption for the $n$-th fading block $\hat{G}_{\textrm{loc}}(\ell_n, \hat{R}_n, h_n)$ has no closed form. {\color{black}{To allow simulation, the function with the properties in Assumption~\ref{Ass:ConvexityG} 
is approximated as 
\begin{equation}
   g( \ell_n)\approx \frac{\gamma \hat{\phi}(\ell_n) \ell_n^3}{T_c^2}
\end{equation}
where $\hat{\phi}(\ell_n)=\frac{(\phi_1-\phi_0)}{(\hat{b}'_n-\hat{b}_n)^4} (\ell -\hat{b}_n)^4+\phi_0$ is a polynomial monotone-increasing and convex function satisfying $\hat{\phi}(\hat{b}_n)=\phi_0$ and $\hat{\phi}(\hat{b}_n^{'})\!=\!\phi_1$.}}
The approximation is verified by simulation to be accurate as shown in
Fig.~\ref{Fig:Approximation performance}.

Fig.~\ref{Fig:Sim:DyChan} (a) and (b) show the curves of computing probability versus deadline $T$ for the Rician factor $K= 0$ and $10$, respectively. The proposed sub-optimal data allocation policy is found to have close-to-optimal performance and substantial performance gain over the equal allocation policy. Moreover, the gain for a highly random channel is larger than that for a LoS channel. This shows that adaptive data allocation is an effective way for  coping  with the effect of fading on mobile cloud computing.

Last, the curves of computing probability versus BS transmission power $P_b$ are plotted in Fig.~\ref{Fig:Sim:DyChan:a} (a) for $K=0$ and in Fig.~\ref{Fig:Sim:DyChan:a} (b) for $K=10$ with a fixed computing deadline  $T = 0.035$ s. Large performance gain is observed for data allocation for the case of highly random channel. Other observations are similar to those from Fig.~\ref{Fig:Sim:DyChan}.

%
%

\section{conclusion}
A novel framework of wirelessly powered mobile cloud computing has been proposed in this paper. Applying optimization theory, a set of policies have been derived for optimizing the computing performance of two mobile operation modes, namely local computing and computation offloading, under the energy harvesting and deadline constraints. Furthermore, given non-causal CSI, a sub-optimal policy for  adaptive data allocation in time has been proposed to cope with the effect of fading on the system computing performance and shown to be close-to-optimal. This set of policies constitute a promising framework for realizing wirelessly powered and cloud-based mobile devices. 

This work can be extended to several interesting directions. First, full-duplex transmission can be implemented in the proposed system to support simultaneous MPT and computation offloading to improve the power transfer efficiency. {\color{black}{Second, the current work focusing on a single-computing task can be generalized to the scenario of computing a multi-task program, which involves program partitioning and simultaneous local computing and offloading.
 Last, it is interesting to extend the current design for single-user mobile cloud computing system to the multiuser system that requires  joint design of radio and computational resource allocation for mobile cloud computing.}}

\appendix

\subsection{Proof of Lemma~\ref{Lem:Relax}}\label{App:Relax}
Define the Lagrangian function for Problem P2 as 
\begin{align}
L&= \sum_{k=1}^N \gamma  p_k f_k^2+ \sum_{m=1}^N \lambda_m\left(\sum_{k=1}^m \gamma f_k^2 - \upsilon P_b h\sum_{k=1}^m y_k\right) \nn \\ 
 &+ \mu \left(\sum_{k=1}^N y_k- T\right) + \sum_{k=1}^N \eta_k\left( \frac{1}{f_k}-y_k\right). 
\end{align} 

Applying the Karush-Kuhn-Tucker (KKT) conditions gives:
\begin{align}
 &\frac{\partial{L}}{\partial f_k^*}\!=\!2 \gamma p_k f_k^*\!+\! 2 \gamma  f_k^*\! \l(\sum_{m=k}^{N} \lambda_m\r)\! \!-\!\eta_k \frac{1}{(f_k^*)^2}=0, \!\!\!\!\!\!& \!\!\!\forall k, \label{Eq: kkt_1} \\ 
&\frac{\partial{L}}{\partial y_k^*}=-\upsilon P_b h \l(\sum_{m=k}^{N} \lambda_m\r)  + \mu -\eta_k =0, &\forall k,    						\label{Eq: kkt_2} \\
&\lambda_m \l [ \sum_{k=1}^m \gamma  (f_k^*)^2 - \upsilon P_b h\sum_{k=1}^m y_k^* \r ]=0,    &\forall m,							\label{Eq: kkt_3}\\
&\mu \(\sum_{k=1}^{N} y_k^* - T\) =0,                   																	\label{Eq: kkt_4}\\
&\eta_k \l(\frac{1}{f_k^*}-y_k^*\r)=0, &\forall k,									
\label{Eq: kkt_5}\\                      
&\lambda_k \ge 0, \mu \ge 0, \eta_k \ge0, f_k^*>0, y_k^*>0,  &\forall k. 	\nn								
 \end{align}
Then it is derived from \eqref{Eq: kkt_1} that 
\begin{equation}\label{Eq:eta:positive}
(f_k^{*})^3 =\frac{\eta_k}{2 \gamma \left(p_k + \sum_{m=k}^N \lambda_m \right)},  ~~\quad \forall k. 
\end{equation}
To ensure $f_k^* >0$, since $\gamma$ and $p_k$ are positive and $\{\lambda_m\}$ are nonnegative, it needs to satisfy that: $\eta_k>0$ for all $k$. Combining it with \eqref{Eq: kkt_5} yields that $y_k^*=\frac{1}{f_k^*}$, leading to the desired result. \hfill $\blacksquare$

\subsection{Proof of Lemma~\ref{Lem:P2}}\label{App:P2}
First, due to the positivity of $\eta_k$ (see Lemma~\ref{Lem:Relax}), it follows from \eqref{Eq: kkt_2} that: $\mu=\eta_k+\upsilon P_b h \l(\sum_{m=k}^N \lambda_k\r)>0$.
Combining it with  \eqref{Eq: kkt_4} gives: $\sum_{k=1}^N y_k^* = \sum_{k=1}^N \frac{1}{f_k^*} =T$.

 Second, the optimal CPU-cycle frequencies \eqref{Eq:Opti:f} can be obtained by combining  \eqref{Eq: kkt_2} and \eqref{Eq:eta:positive}. 
 
 Third, compare $f_k^*$ and $f_{k+1}^*$ based on \eqref{Eq:Opti:f}. Since $p_{k+1} < p_k$ and $\sum_{m=k+1}^{N} \lambda_m \le \sum_{m=k}^N \lambda_m$, it follows that $f_{k+1}^* > f_k^*$, completing the proof. \hfill $\blacksquare$
\vspace{-5pt}

\subsection{Proof of Lemma~\ref{Lem:RelaxMore}}\label{App:RelaxMore}
It is proved by contradiction as follows. Assume there exists an integer $m$ where $2\le m \le N-1$ such that $\lambda_m>0$.
It follows from \eqref{Eq: kkt_3} that
\begin{align*}
&\sum_{k=1}^m \left[\gamma  (f_k^*)^2 - \upsilon P_b h \frac{1}{f_k^*}\right]\\
=& \sum_{k=1}^{m-1} \left[\gamma  (f_k^*)^2 - \upsilon P_b h\frac{1}{f_k^*}\right]+\left[ \gamma  (f_m^*)^2 - \upsilon P_b h\frac{1}{f_m^*}\right]=0.
\end{align*} 
Since $\sum_{k=1}^{m-1} \left[\gamma  (f_k^*)^2 - \upsilon P_b h\frac{1}{f_k^*}\right] \le 0$, it can be obtained that $ \gamma  (f_m^*)^2 - \upsilon P_b h\frac{1}{f_m^*} \ge 0.$
Combining it with the monotonicity of $\{f_k^* \}$ (see Lemma~\ref{Lem:P2} ), for the $(m+1)$-th cycle, it has $\gamma  (f_{m+1}^*)^2 - \upsilon P_b h\frac{1}{f_{m+1}^*} >\gamma  (f_m^*)^2 - \upsilon P_b h\frac{1}{f_m^*} \ge 0,$
which results in
\begin{align*}
&\sum_{k=1}^{m+1} \left[\gamma  (f_k^*)^2 - \upsilon P_b h \frac{1}{f_k^*}\right]\\
=& \sum_{k=1}^{m} \left[\gamma  (f_k^*)^2 - \upsilon P_b h\frac{1}{f_k^*}\right]+\left[ \gamma  (f_{m+1}^*)^2 - \upsilon P_b h\frac{1}{f_{m+1}^*}\right]>0,
\end{align*}
 contradicting the energy harvesting constraint. Similar proof by contradiction also applies to $\lambda_1$. Therefore, the Lagrange multipliers $\{\lambda_m\}$ for Problem P2 satisfy $\lambda_1=\lambda_2=\cdots=\lambda_{N-1}=0$ and $\lambda_N \ge 0$, yielding the desired result. \hfill $\blacksquare$

\subsection{Proof of Theorem~\ref{Theo:LocalComp}}\label{App:LocalComp}
Considering Problem P3, the conditions for feasible cases are derived as follows.
\begin{enumerate}
\item{
 \emph{Case 1}: $\lambda>0$. 
First, substituting  \eqref{reduced optimal f} into the deadline constraint leads to 
 \begin{equation}
 \left( \frac{2 \gamma}  {\mu-\upsilon P_b h \lambda}  \right)^ {\frac{1}{3}} \left[ \sum_{k=1}^N  (p_k+\lambda) ^{\frac{1}{3}}\right]=T. \label{Eq:Deadline:Loc}
\end{equation}
Combining \eqref{reduced optimal f} and \eqref{Eq:Deadline:Loc} gives \eqref{Eq:CPU:lambda:nonzero}.
Next, substituting  \eqref{reduced optimal f} into the energy harvesting constraint  yields
 \begin{equation}\label{Eq:EnergyCausality1}
 \left( \frac{\mu-\upsilon P_b h \lambda}{2 \gamma}\right)^{\frac{2}{3}} \left[ \sum_{k=1}^N \l({p_k+\lambda}\r)^{-\frac{2}{3}}\right]=\frac{\upsilon P_b h T}{\gamma}.
\end{equation}
Combining \eqref{Eq:Deadline:Loc} and \eqref{Eq:EnergyCausality1} results in \eqref{Eq:lambda:ph}.
Then, consider the monotone property of $P_b h$. 
From \eqref{Eq:lambda:ph}, the first derivative of $P_b h$ w.r.t $\lambda$ is
\begin{align*}
\frac{\partial (P_b h)}{\partial \lambda}&=\frac{2 \gamma}{3 \upsilon  T^3} \left(  \sum_{k=1}^{N} u_k\right) \\
& \times  \left[   \left( \sum_{k=1}^{N} v_k  \right)^2\!-\! \left( \sum_{k=1}^{N} v_k w_k \right)\left( \sum_{k=1}^{N}\frac{v_k}{w_k} \right) \right]
\end{align*}  where $u_k\!=\!(p_k+\lambda)^{\frac{1}{3}}$, $v_k=(p_k + \lambda)^{\frac{-2}{3}}$ and $w_k=p_k + \lambda$.
Applying Cauchy inequality gives that: $\left( \sum_{k=1}^{N} v_k w_k \right)\left( \sum_{k=1}^{N}\frac{v_k}{w_k} \right) \overset{(a)}{\ge} 
 \left( \sum_{k=1}^{N} v_k  \right)^2.$
The equality in $(a)$ holds only when $\{\sqrt{v_k w_k} \l / \r. \sqrt{v_k/w_k }\}$ are equal for all $k$, which cannot be satisfied in this problem. Consequently, $\frac{\partial P_b}{\partial \lambda}<0.$ 
Furthermore, the asymptotic properties of $P_b h$ w.r.t $\lambda$ are characterized: when $\lambda \rightarrow 0$, it has $P_b h \rightarrow a'$; when $\lambda \rightarrow \infty$,  it has $P_b h \rightarrow a$ where $a$ and $a'$ are defined in \eqref{Eq:a and a'}. Combining them with the monotone-decreasing property of $P_b h$ gives that: if $0<\lambda< \infty$,  it leads to $a< P_b h<a'$. In addition, when $P_b h=a$, there is only one feasible solution: $f_k^*=\frac{N}{T}$ for all $k$, which can also be expressed as \eqref{Eq:CPU:lambda:nonzero} by letting $\lambda = \infty$. 
}
\item{
\emph{Case 2}: $\lambda=0$. Since \eqref{Eq:CPU:lambda:nonzero} is derived only using the deadline constraint,  \eqref{Eq:CPU:lambda:zero} can be obtained by letting $\lambda=0$ in \eqref{Eq:CPU:lambda:nonzero}.
Substituting it into the energy harvesting constraint ($\sum_{k=1}^N \gamma  (f_k^*)^2 \le \upsilon P_b h T$) gives: $P_b h\ge a'$.
}
\end{enumerate}
Thus, it can be concluded that if $P_b h<a$, Problem P3 is infeasible, completing the proof. \hfill $\blacksquare$

\subsection{Proof of Corollary~\ref{Cor:Egy:Local}}\label{App:Egy:Local}
The results of \eqref{Eq:Local:egy:lambda:nonzero}, \eqref{Eq:Local:egy:lambda:zero} and asymptotic properties can be derived straightforwardly following Theorem~\ref{Theo:LocalComp}. The monotone property of $\bar{E}_{\textrm{loc}}^*$ for the case of $a \le P_b h<a'$ is proved as follows.

For notation simplicity, define $a_k=(p_k+\lambda)^{\frac{-1}{3}}$ such that $p_k\!=\! a_k^{-3} -\lambda$. From \eqref{Eq:Local:egy:lambda:nonzero}, the first derivative of $\bar{E}_{\textrm{loc}}^*$ w.r.t $\lambda$ is:
\begin{align}
&\frac{\partial \bar{E}_{\textrm{loc}}^*}{\partial \lambda}\!=\!\frac{2\gamma}{3 T^2} \left(  \sum_{k=1}^{N} a_k^{-1}\right) 
 \left[   \left( \sum_{k=1}^{N} a_k^2  \right)\left( \sum_{k=1}^{N} (a_k^{-3}-\lambda) a_k^2 \right)\right.\nn\\ 
& \left. -\left( \sum_{k=1}^{N} a_k^{-1} \right)\left( \sum_{k=1}^{N}(a_k^{-3}-\lambda) a_k^5 \right)  \right]. \nn
\end{align}
By algebraic calculation, the part in the square bracket  is
\begin{align*}
&\sum_{i=1}^{N} \sum_{j=1, j\neq i}^{N} \frac{1}{a_i a_j} \left[a_i^3-a_j^3+\lambda(a_j^6-a_i^3 a_j^3)\right]\\
=&\sum_{i=1}^{N} \sum_{j= i+1}^{N} \frac{1}{a_i a_j} \left[(\lambda (a_i^3-a_j^3)^2 \right] > 0, \nn
\end{align*}
leading to $\frac{\partial \bar{E}_{\textrm{loc}}^*}{\partial \lambda} >0$. Combing it with $\frac{\partial (P_b h)}{\partial \lambda} <0$ and one-one mapping between $P_b h$ and $\lambda$ results in that $\bar{E}_{\textrm{loc}}^*$ is a monotone-decreasing function  of $P_b h$. 
\hfill $\blacksquare$

\subsection{Proof of Lemma~\ref{Lem:Conv:Off}}\label{App:Conv:Off}
Define two constants for Problem P4: $d=\frac{\sigma^2}{h}-\upsilon P_b h$ and $d'=-\frac{\sigma^2}{h}$. Then the first derivative of  $S_{\textrm{off}}$ w.r.t $t$  for $t \in(0, \infty)$ is given as
\begin{equation}
\frac{\partial {S_{\textrm{off}}}}{\partial t} = d+\l(d'-\frac{d'L\ln 2}{Bt}\r) 2^{\frac{L}{Bt}}. \label{Eq:FirstOrder1}
\end{equation}
The second derivative follows:
\begin{equation}
\frac{\partial^2 {S_{\textrm{off}}} } {\partial t^2}= \frac{d' L^2 ({\ln 2})^2}{B^2 t^3} 2^{\frac{L}{Bt}} <0,
\end{equation}
since $d' <0$, verifying the concavity of $S_{\textrm{off}}$.
 Then letting the first derivative \eqref{Eq:FirstOrder1} be  zero gives:
\begin{equation}
\frac{d}{d'}=\left(\frac{L\ln2}{Bt}-1\right)2^{\frac{L}{Bt}}. \label{Eq:FirstOrder2}
\end{equation}
Using the Lambert function, the solution for \eqref{Eq:FirstOrder2} is: $t= \rho(h) L$ with $\rho(h)$ defined in \eqref{Eq:rho}.
Furthermore, it can be observed from \eqref{Eq:FirstOrder1} that if $t\rightarrow 0$, then $\frac{\partial {S_{\textrm{off}}}}{\partial t}\rightarrow \infty.$ Therefore, $S_{\textrm{off}}$ is maximized at $t=\rho(h) L$, leading to the desired result.
\hfill $\blacksquare$

\subsection{Proof of Theorem~\ref{Theo:Offload}}\label{App:Offload}
Based on Lemma~\ref{Lem:Conv:Off}, if $\rho(h) L\ge T$,  $S_{\textrm{off}}^*$ is maximized at $t=T$ and $S_{\textrm{off}}^*<0$ such that Problem P4 is infeasible. Therefore, to guarantee the feasibility of Problem P4, two conditions should be satisfied: 1) $t^*=\rho(h) L < T$; 2) $S_{\textrm{off}}^*(t^*)\ge 0$.

  First, since $\frac{\partial {S_{\textrm{off}}}}{\partial t}\rightarrow \infty$ when $ t\rightarrow 0$, it only needs to satisfy that  when $t= T$, it has $\frac{\partial {S_{\textrm{off}}}}{\partial t}<0.$ From \eqref{Eq:FirstOrder1}, it can be obtained that
\begin{equation} 
P_b h^2>\frac{\sigma ^2}{\upsilon}\left[\(\frac{L\ln2}{BT}-1\)2^{\frac{L}{BT}}+1\right]. \label{Eq:1stcondition:offloading}
\end{equation}

 Next, substituting $t^*$ satisfying \eqref{Eq:FirstOrder2} into \eqref{Eq:Offloading:Objective}
and letting $S_{\textrm{off}}^* \ge 0$  gives
\begin{equation}\label{Eq:Positive_threshold_power}
\frac{\upsilon P_b h^2}{ \sigma ^2 e}\ge\frac{L \ln 2}{BT }\exp \l({W\left(\frac{\upsilon P_b h^2}{\sigma ^2 e} -\frac{1}{e}\right)} \r).
\end{equation}

Denote $d''=\frac{\upsilon P_b h^2}{ \sigma ^2 e}-\frac{1}{e}$ and $d'''=\frac{L \ln 2}{BT }$. Then \eqref{Eq:Positive_threshold_power} is rewritten as $d''+\frac{1}{e}\ge d''' e^{W(d'')}.$

Applying $d''=W(d'')e^{W(d'')}$ to the above inequality and multiplying $e^{-d'''}$  on both sides gives: $\left[ W(d'')-d''' \right]e^{W(d'')-d''' } \ge -e^{-1-d'''}.$
Then, it follows that
\begin{equation}
d'' \ge \left[ d'''+W(-e^{-1-d'''})\right] e^{d'''+W(-e^{-1-d'''})}. \label{Eq:d''}
\end{equation}
Substituting the expression of $d''$ and $d'''$ to \eqref{Eq:d''} gives the solution for \eqref{Eq:Positive_threshold_power} as $P_b h^2 \ge a''$
where $a''$ is defined in \eqref{Eq:Thresh:Off}.

Last, combine the two conditions \eqref{Eq:1stcondition:offloading} and $P_b h^2 \ge a''$. 
Since $ -1< W(-e^{-1-d'''})<0$, it has $a'' > \frac{\sigma ^2}{\upsilon}\left[\(\frac{L\ln2}{BT}-1\)2^{\frac{L}{BT}}+1\right]$. In conclusion, $P_b h^2 \ge a'' $, completing the proof.  \hfill $\blacksquare$

\subsection{Proof of Lemma~\ref{Lem:PositiveResidual}}\label{App:PositiveResidual}
According to Corollary~\ref{Cor:DataThreshold},
the mobile obtains the minimum average energy savings in this block if $\ell_n=b_n'$. In this case, 
 the residual energy for the next fading block is
\begin{align}
R_{n+1}&=R_n+\upsilon P_b h_n T_c - G_{\textrm{loc}}(b_n', R_n, h_n)\\
&=R_n+\upsilon P_b h_n T_c-\frac{\gamma \phi_1 {b_n'}^3}{T_c^2}. \label{Eq:Min:ReEgy}
\end{align}
Then substituting $b_n'$ given in \eqref{Eq:Const:L} into \eqref{Eq:Min:ReEgy} gives the lower bound. The upper bound is achieved when $\ell_n=0$.
\hfill $\blacksquare$
\vspace{-10pt}

\subsection{Proof of Proposition~\ref{Pro:Data:Local}}\label{App:Data:Local}
Define the Lagrangian function for Problem P7: 
\begin{align}
L&=\sum_{n=1}^M \hat{G}(\ell_n, \hat{R}_n, h_n)+ \xi \l(L-\sum_{n=1}^M \ell_n\r)\\
&+\sum_{n=1}^M \varpi_n \l(-\ell_n\r)+\sum_{n=1}^M \zeta_n \l(\ell_n-\hat{b}'_n\r). \nn
\end{align}
Applying the KKT conditions leads to
\begin{align}
 &\frac{\partial{L}}{\partial \ell_n}\!=\! \frac{\partial{\hat{G}}}{\partial \ell_n}(\ell_n^*, \hat{R}_n, h_n)-\xi-\varpi_n+\zeta_n=0, & \forall n, \label{Eq: kkt2_1} \\ 
&\varpi_n \ell_n^*=0, ~ \zeta_n \l(\ell_n^*-\hat{b}'_n\r)\!=\!0, ~ \varpi_n \ge 0, ~ \zeta_n\ge 0,  &\forall n,							\label{Eq: kkt2_2}\\
&\sum_{n=1}^M \ell_n^*=L. 						\nn			\label{Eq: kkt2_4}                    
 \end{align}

 First, it can be proved that $\ell_n^*>0$ and $\varpi_n=0$ for all $n$ by the following steps:
 \begin{enumerate}
 \item{ Observe from \eqref{Eq:UniformG} that $ \frac{\partial{\hat{G}}}{\partial \ell_n}(\ell_n^*, \hat{R}_n, h_n)\ge 0$ and equals to $0$ only when $\ell_n^*=0$.
 }
 \item{  Suppose there exists a $n$ such that $\ell_n^*=0$. It leads to $\zeta_n=0$ and $ \frac{\partial{\hat{G}}}{\partial \ell_n}(\ell_n^*, \hat{R}_n, h_n)=0$. From \eqref{Eq: kkt2_1}, it gives $\xi=-\varpi_n\le 0$.}
 \item{ There always exists one  $j$ where $j \neq n$ such that $\ell_j^* >0$. Therefore, $\varpi_j=0$ and $\frac{\partial{\hat{G}}}{\partial \ell_j}(\ell_j^*, \hat{R}_j, h_j)>0$. From \eqref{Eq: kkt2_1}, it can be derived that $\zeta_j=\xi-\frac{\partial{\hat{G}}}{\partial \ell_j}(\ell_j^*, \hat{R}_j, h_j)<0$ which contradicts to the condition $\zeta_j \ge 0$ and leads to the conclusion.
 }
 \end{enumerate}
 
 Then, the  data allocation follows:
 \begin{enumerate}
 \item{If $\ell_n^* < \hat{b}'_n$, then $\zeta_n=0$ and $ \frac{\partial{\hat{G}}}{\partial \ell_n}(\ell_n^*, \hat{R}_n, h_n)=\xi$.}
 \item{If $\ell_n^* = \hat{b}'_n$,  then $$ \frac{\partial{\hat{G}}}{\partial \ell_n}(\ell_n^*, \hat{R}_n, h_n)+\zeta_n=\frac{\partial{\hat{G}}}{\partial \ell_n}(\hat{b}'_n, \hat{R}_n, h_n)+\zeta_n=\xi$$.}
 \end{enumerate}
 
 In conclusion, $\ell_n^*=\min\{b_n(\xi), \hat{b}'_n \}$ where $b_n(\xi)$ is the root of function $\frac{\partial{\hat{G}}}{\partial \ell_n}(b_n, \hat{R}_n, h_n)=\xi$ and $\xi$ satisfies $\sum_{n=1}^M \ell_n^*=L$. Specifically, when $b_n(\xi)\ge \hat{b}'_n=\l(\frac{\upsilon P_b h_n T_c^3 + \hat{R}_n T_c^2 }{\gamma \theta_1}\r)^{\frac{1}{3}}$, it has $h_n \le\frac{\gamma \theta_1 b_n^3(\xi)-\hat{R}_n T_c^2}{\upsilon P_b T_c^3}$, completing the proof.
\hfill $\blacksquare$

\bibliographystyle{ieeetr}

\vspace{-10pt}
\begin{IEEEbiography}
[{\includegraphics[width=1in,clip,keepaspectratio]{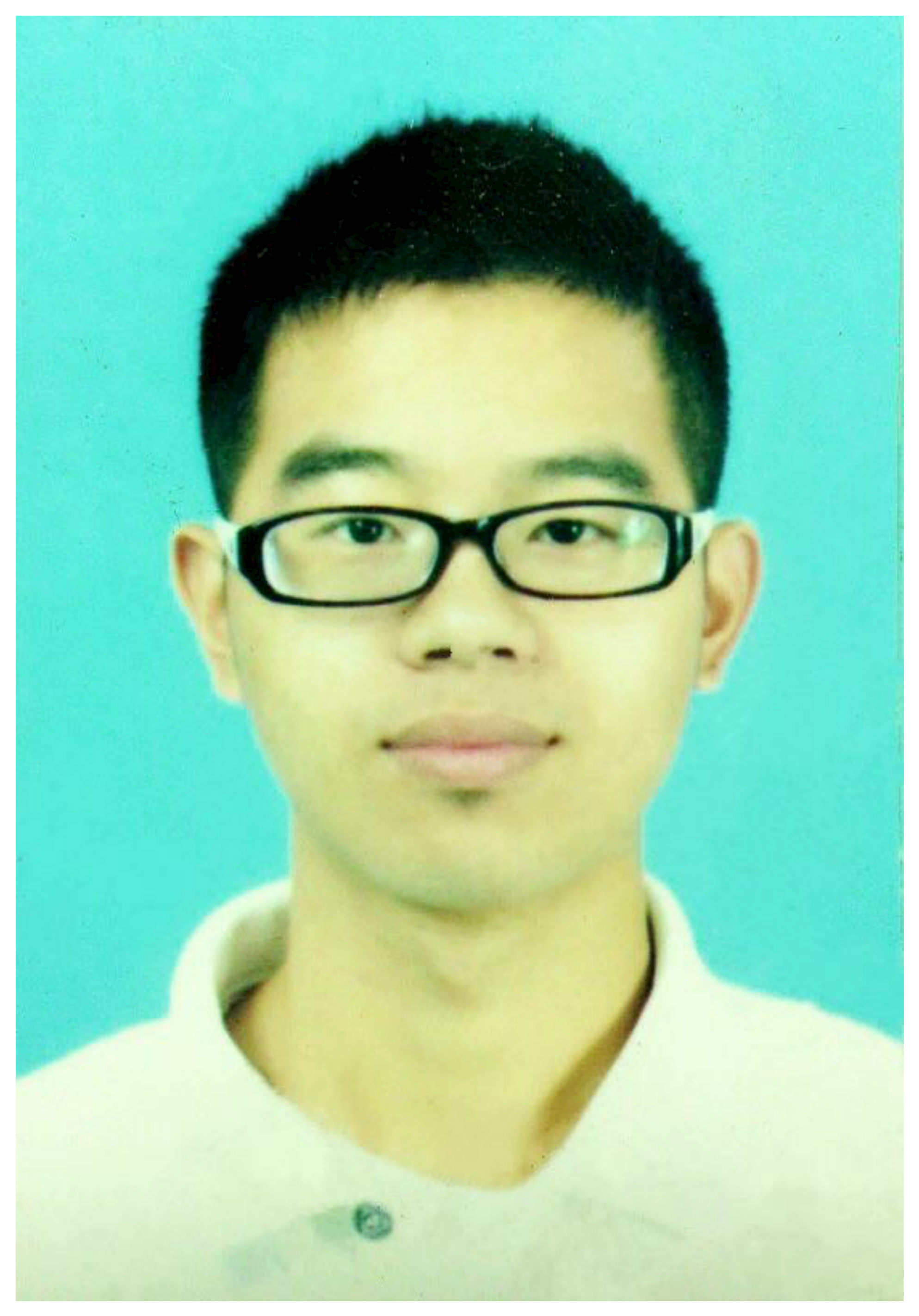}}]
{Changsheng You}(S'15) received the B.S. degree in electronic engineering and information science from the University of Science and Technology of China (USTC) in 2014. He is currently working towards the Ph.D. degree in electrical and electronic engineering at The University of Hong Kong (HKU). His research interests include mobile cloud computing, wireless power transfer, energy harvesting system and convex optimization.
\end{IEEEbiography}
\vspace{-10pt}
\begin{IEEEbiography}
[{\includegraphics[width=1in,clip,keepaspectratio]{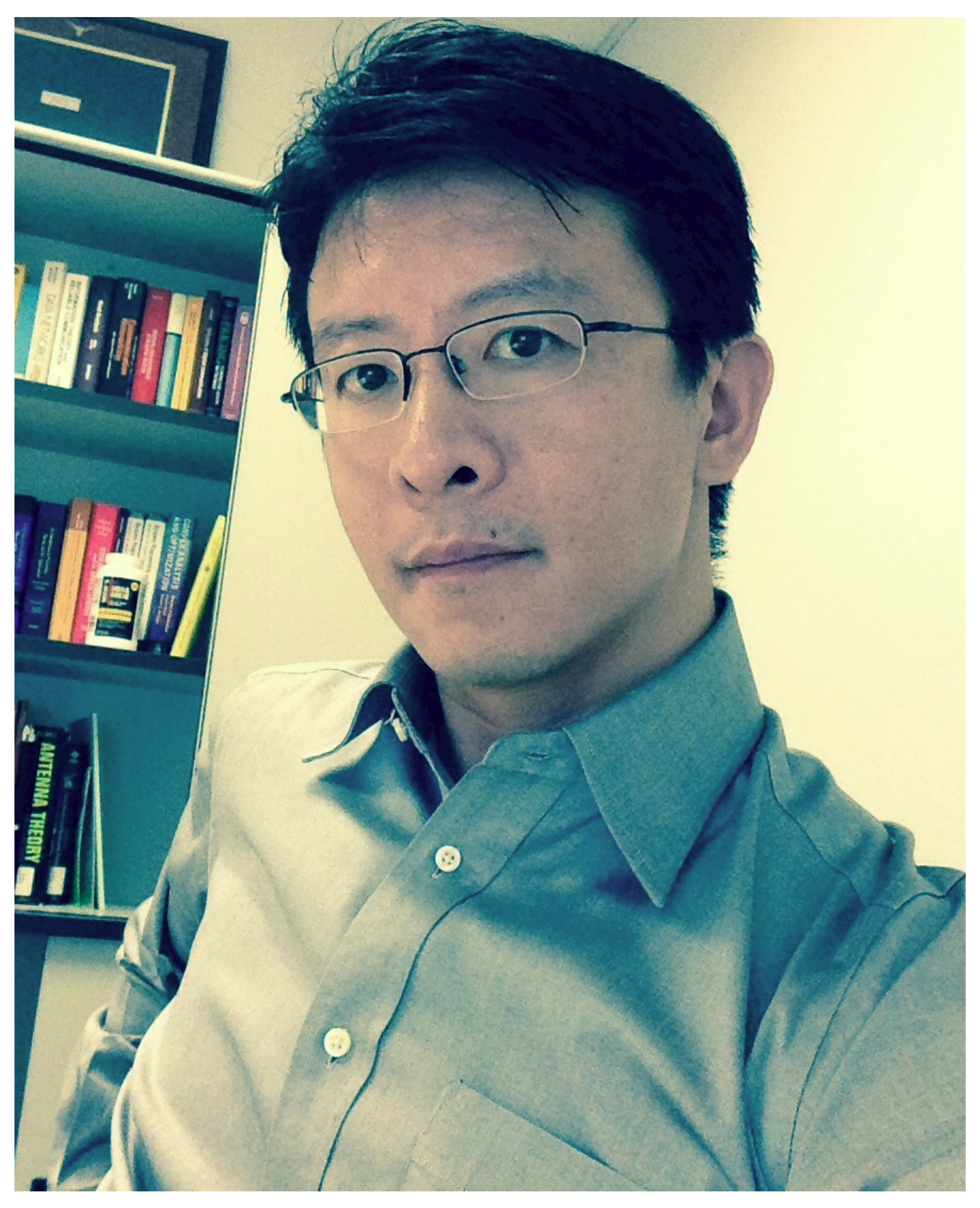}}]
{Kaibin Huang} (M'08-SM'13) received the B.Eng. (first-class hons.) and the M.Eng. from the National University of Singapore, respectively, and the Ph.D. degree from The University of Texas at Austin (UT Austin), all in electrical engineering.

Since Jan. 2014, he has been an assistant professor in the Dept. of Electrical and Electronic Engineering (EEE) at The University of Hong Kong. He is an adjunct professor in the School of EEE at Yonsei University in S. Korea. He used to be a faculty member in the Dept. of Applied Mathematics (AMA) at the Hong Kong Polytechnic University (PolyU) and the Dept. of EEE at Yonsei University. He had been a Postdoctoral Research Fellow in the Department of Electrical and Computer Engineering at the Hong Kong University of Science and Technology from Jun. 2008 to Feb. 2009 and an Associate Scientist at the Institute for Infocomm Research in Singapore from Nov. 1999 to Jul. 2004. His research interests focus on the analysis and design of wireless networks using stochastic geometry and multi-antenna techniques.

He frequently serves on the technical program committees of major IEEE conferences in wireless communications. He has been the technical chair/co-chair for the IEEE CTW 2013, the Comm. Theory Symp. of IEEE GLOBECOM 2014, and the Adv. Topics in Wireless Comm. Symp. of IEEE/CIC ICCC 2014 and has been the track chair/co-chair for IEEE PIMRC 2015, IEE VTC Spring 2013, Asilomar 2011 and IEEE WCNC 2011. Currently, he is an editor for IEEE Journal on Selected Areas in Communications (JSAC) series on Green Communications and Networking, IEEE Transactions on Wireless Communications, IEEE Wireless Communications Letters. He was also a guest editor for the JSAC special issues on communications powered by energy harvesting and an editor for IEEE/KICS Journal of Communication and Networks (2009-2015). He is an elected member of the SPCOM Technical Committee of the IEEE Signal Processing Society. Dr. Huang received the 2015 IEEE ComSoc Asia Pacific Outstanding Paper Award, Outstanding Teaching Award from Yonsei, Motorola Partnerships in Research Grant, the University Continuing Fellowship from UT Austin, and a Best Paper Award from IEEE GLOBECOM 2006 and PolyU AMA in 2013. 
\end{IEEEbiography}
\vspace{-10pt}
\begin{IEEEbiography}
[{\includegraphics[width=1in,clip,keepaspectratio]{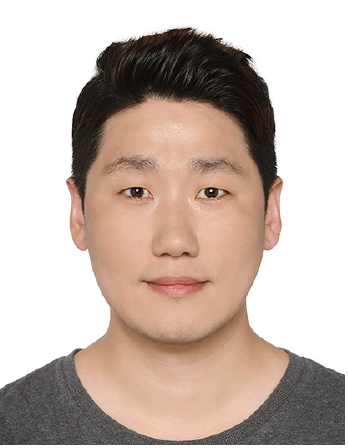}}]
{Hyukjin Chae} received the B.S and Ph.D degree in electrical and electronic
engineering from Yonsei University, Seoul, Korea. He joined LG Electronics,
Korea, as a Senior Research Engineer in 2012. His research interests
include interference channels, multiuser MIMO, D2D, V2X, and full duplex
radio. From Sep. 2012, he has contributed and participated as a delegate in
3GPP RAN1 with interests in ePDCCH, eIMTA, FD MIMO, Indoor positioning,
D2D, and V2X communications. He is an inventor of more than 100 patents.
\end{IEEEbiography}
\end{document}